\documentclass[twoside]{article}

\usepackage{jmlr2e}
\usepackage{hyperref}
\hypersetup{
    colorlinks=true,
    linkcolor=blue,
    filecolor=red,      
    urlcolor=blue,
    citecolor=blue,
    linktocpage
}

\usepackage{comment}

\usepackage{indentfirst}

\usepackage{bigints}
\usepackage{nccmath} 
\usepackage{xpatch}
\xpatchcmd{\NCC@ignorepar}{%
\abovedisplayskip\abovedisplayshortskip}
{%
\abovedisplayskip\abovedisplayshortskip%
\belowdisplayskip\belowdisplayshortskip}
{}{}

\usepackage{showidx}
\usepackage{tikz}
\usepackage{optidef}
\usepackage{graphicx}
\usepackage{xcolor}
\usepackage{hyperref}
\usepackage[english]{babel}

\usepackage{hyperref}

\usepackage[linesnumbered,ruled,vlined, commentsnumbered]{algorithm2e}

\usepackage{tabularx, floatrow, ragged2e, booktabs, caption}
\usepackage{verbatim}
\usepgflibrary{arrows}
\usepackage{verbdef}
\usepackage{color}
\usepackage{graphicx}
\usepackage[european]{circuitikz}
\usepackage{arydshln}
\usepackage[T1]{fontenc}
\usepackage[thinlines]{easytable}

\usepackage{mathtools}
\usepackage[amsmath]{empheq}

\usepackage{doi}

\tikzset{>=latex}

\tikzset{%
block/.style    = {draw, thick, rectangle, minimum height = 3em,
    minimum width = 3em},
  block1/.style    = {draw, thick, rectangle, minimum height = 1.7em,
    minimum width = 1.7em,fill=gray!70},
      block2/.style    = {draw, thick, rectangle, minimum height = 1.7em,
    minimum width = 1.7em,fill=gray!30},
  sum/.style      = {draw, circle, node distance = 1.8cm}, 
}

\firstpageno{1}

\renewcommand{\tilde}{\widetilde}

\renewcommand{\hat}{\widehat}
\newcommand{\FF}{{{\rm I \kern -0.2em R}}}
\newcommand{\RR}{{{\rm I \kern -0.2em R}}}

\newcommand{\CC}{{{\mbox{\rm \hspace*{0.05ex}
\rule[.18ex]{.18ex}{1.24ex} \kern -.65em C}}}} 

\newcommand{\bea}{\begin{eqnarray}}
\newcommand{\eea}{\end{eqnarray}}

\newtheorem{thm}{Theorem}[section]
\newtheorem{rem}{Remark}
\newtheorem{prop}[thm]{Proposition}
\newtheorem{lem}[thm]{Lemma}

\newtheorem{defn}[thm]{Definition}
\newtheorem{assumption}{Assumption}

\newcommand{\ba}{\left[ \begin{array}}
\newcommand{\baa}{\begin{array}}
\newcommand{\ea}{\end{array} \right]}
\newcommand{\eaa}{\end{array}}

\newcommand{\be}{\begin{equation}}
\newcommand{\ee}{\end{equation}}
\newcommand{\bb}{\begin{equation}\label}

\newcommand{\thmref}[1]{Theorem~\ref{#1}}

\newcommand{\lemref}[1]{Lemma~\ref{#1}}

\newcommand{\remref}[1]{Remark~\ref{#1}}
\newcommand{\defref}[1]{Definition~\ref{#1}}
\newcommand{\propref}[1]{Proposition~\ref{#1}}

\newcommand{\figref}[1]{Figure~\ref{#1}}
\def\math#1{\ifmmode{#1} \else {$#1$}\fi}

\newcommand{\sg}{\ifmmode \Sigma \else $\Sigma$ \fi}

\date{}

\begin{document}

\title{Sample Complexity of the Robust LQG Regulator \\  with Coprime Factors Uncertainty}

\author{Yifei Zhang$^\sharp$,\thanks{$^\sharp$Yifei Zhang
is with the Electrical and Computer Engineering Deptartment,
Stevens Institute of Technology, Hoboken, NJ 07030 USA.
email: yzhang133@stevens.edu}
Sourav Kumar Ukil$^\ddag$ \thanks{$^\sharp$Sourav Kumar Ukil
is with the Electrical and Computer Engineering Deptartment,
Stevens Institute of Technology, Hoboken, NJ 07030 USA.
email: sukil@stevens.edu}
and Serban Sabau$^\flat$ \thanks{$^\sharp$\c{S}erban Sab\u{a}u
is with the Electrical and Computer Engineering Deptartment,
Stevens Institute of Technology, Hoboken, NJ 07030 USA.
email: ssabau@stevens.edu}}

\author{\name Yifei Zhang \email yzhang133@stevens.edu \\
       \addr Department of Electrical and Computer Engineering\\
       Stevens Institute of Technology,
       NJ 07030, USA
       \AND
       \name Sourav Kumar Ukil \email sukil@stevens.edu \\
       \addr Department of Electrical and Computer Engineering\\
       Stevens Institute of Technology,
       NJ 07030, USA
       \AND
       \name Ephraim Neimand \email eneimand@stevens.edu \\
       \addr Department of Electrical and Computer Engineering\\
       Stevens Institute of Technology,
       NJ 07030, USA
       \AND
       \name \c{S}erban Sab\u{a}u \email ssabau@stevens.edu \\
       \addr Department of Electrical and Computer Engineering\\
       Stevens Institute of Technology, 
       NJ 07030, USA
       \AND
       \name Myron E. Hohil \email myron.e.hohil.civ@mail.mil \\
       \addr 
U.S. Army Combat Capabilities Development Command \\ (CCDC)
Armaments Center,  Picatinny Arsenal, NJ 07806, USA}


\maketitle

\begin{abstract}
This paper addresses the end-to-end sample complexity bound for learning the $\mathcal{H}_2$ optimal controller (the Linear Quadratic Gaussian ({\bf LQG}) problem) with unknown dynamics, for potentially {\em unstable} Linear Time Invariant ({\bf LTI}) systems. The robust LQG synthesis procedure is performed by considering bounded additive model uncertainty on the coprime factors of the plant. The closed-loop identification of the nominal model of the true plant is performed  by constructing a Hankel-like matrix from a single time-series of noisy finite length input-output data, using the ordinary least squares algorithm from \cite{Dahleh2020}. Next, an $\mathcal{H}_{\infty}$ bound on the estimated model error is provided and the robust controller is designed via convex optimization, much in the spirit of  \cite{Boczar2018} and \cite{furieri2020}, while  allowing for bounded additive uncertainty on the coprime factors of the model. 
Our conclusions are consistent with previous results on learning the LQG and LQR controllers.
\end{abstract}
$ \fontdimen14\textfont2=6pt
\fontdimen16\textfont2=2.5pt$
\section{Introduction}

Considerable research efforts have been spent within the last few years towards approaching classical control problems with modern statistical and optimization tools from the Machine Learning framework, envisaging  practical applications, see for example \cite{Dean2018}, \cite{mania2019}, \cite{dean2020}, \cite{furieri2020}.
The starting point of the aforementioned research efforts has been the classical LQG control problem,  which deals with partially observed linear and time-invariant dynamical systems driven by Gaussian noise and where the problem is finding the optimal output feedback law that minimizes the expected value of a quadratic cost.

In this paper an end-to-end sample-complexity bound on learning LQG controllers that stabilize the true system with  high probability is established by incorporating  recent advances in  finite time (non-asymptotic) system identification (\cite{Dahleh2020}). The contribution resides in the development of a tractable robust control synthesis procedure, that allows for   bounded additive model uncertainty on the coprime factors of the model of the plant, thus circumventing both the need for state feedback and the restrictive assumption on the plant's open loop stability.  The resulted sub-optimality gap is bounded as a function of the level of the model uncertainty. The  end-to-end sample complexity bound for learning robust LQG controllers is $\mathcal{O}(\sqrt{{logT}/{T}})$, where $T$ is the time horizon for learning. For open-loop stable systems, \cite{furieri2020} recently proved that the performance for LQG controllers deteriorates linearly with the model estimation error, starting from the original analysis of \cite{Dean2018} from the case of learning fully observed LQR controllers. The  robust control synthesis proposed here achieves the same scaling for the sub-optimality gap as \cite{Dean2018}, namely $\mathcal{O}(\gamma^2)$, where $\gamma$ is the model uncertainty level.


\subsection{The Linear Quadratic Gaussian Problem}

Within the last few years, modern statistical and algorithmic methods led to new solutions for classical control problems, such as the Linear Quadratic Gaussian problem. For a discrete-time {\bf  LTI} (Linear and Time Invariant) systems driven by Gaussian process and sensor noise:
\useshortskip
\begin{equation}\label{stateeq}
\begin{aligned}
      x_{k+1} & = {A} x_k + {B} u_k + {\delta x}_k,\\  y_k & = {C} x_k + {D} u_k + \nu _k,
\end{aligned}
\end{equation}
\noindent
where $x_k \in \mathbb{R}^n$ is the state of the system, $u_k \in \mathbb{R}^m$ is the control input and $y_k \in \mathbb{R}^p$ is the measurement output with ${\delta x}_k \in \mathbb{R}^m$, $\nu_k \in \mathbb{R}^p$ are Gaussian noise with zero mean, covariance $\sigma_{\delta x}^2 I$ and $\sigma_\nu^2 I$ 
respectively, the classical LQG control problem is defined as:
\noindent
\useshortskip
\begin{equation}\label{LQGCostOriginal}
\begin{aligned}
\min_{u_0, u_1, ...}  \quad  \lim\limits_{T \rightarrow \infty} &  \mathbb{E} \bigg[ \dfrac{1}{T} \sum \limits_{t=0}^{T} (y^T_t P_1 y_t + u^T_t P_2 u_t)  \bigg] \\
\textrm{subject to} \quad & \eqref{stateeq},
\end{aligned}
\end{equation}
where, $P_1, P_2$ is positive definite. Without loss of generality, it is assumed that $P_1 = I_p$, $P_2 = I_m$, $\sigma_{\delta x} = 1$, $\sigma_\nu = 1$.

In a nutshell, the problem can be stated as {\em learning}  with high probability and in finite time the model of an unknown LTI system and subsequently designing its optimal LQG controller, while accounting for the inherent model uncertainty incurred at the {\em learning} stage.

\subsection{The Main Technical Ingredient}

Identifying LTI models from input-output data has been the focus of time-domain identification. 
Using coprime factors instead of the state space realization of the system has a great advantage as it ensures to work on unstable system identification. The Transfer Function Matrix (\textbf{TFM}) of the plant is written as: 
\begin{equation}\label{plantequation}
    \textbf{G}(z) = \tilde{\textbf{M}}^{-1}(z) {\bf \tilde{N}}(z) = \textbf{N}(z)\textbf{M}^{-1}(z) 
\end{equation}
\noindent
where $\tilde{\textbf{M}}(z)$, $\tilde{\textbf{N}}(z)$, $\textbf{M}(z)$ and $\textbf{N}(z)$ are stable TFM. 
Doubly coprime factorization of the TFM of a LTI system plays a key role in many sectors of the factorization approach to filter synthesis and multi-variable control systems analysis. A doubly coprime factorization of a given LTI plant is closely related to the Parameterization of all stabilizing controllers for this plant.
With a realization of the TFM, various formulas to compute doubly coprime factorizations over the ring of stable and proper Rational Matrix Functions (\textbf{RMF}) have been proposed both for standard (proper) and for generalized (improper, singular, or descriptor) systems. 
The formulas are expressed either in terms of a stabilizable and detectable realization of the underlying TFM or make additional use of a realization of a full or reduced order observer based stabilizing controller.

\subsection{Contributions}

Recently, LQG control has been studied  in a  model based Reinforcement Learning framework (\cite{furieri2020}) and the sub-optimality performance degradation of the robust LQG controller was proved to scale as a function of the modeling error. However, the results in \cite{furieri2020} are valid only for open-loop stable systems, thus excluding many situations of practical interest. This paper shows how to remove the stability assumption on the unknown system, while at the same time streamlining the equivalent optimization problem, by reducing the size of the subsequent linear constraints.  The proposed algorithm  is consistent with previous results, while allowing for a much stronger description of the modeling error as bounded additive uncertainty on the coprime factors of the model of the plant (without restriction on the McMillan degree of the true plant or on its number of unstable poles). As expected, the presence of additive, norm-bounded factors on the coprimes of the plant renders the cost functional  non-convex, therefore the  derivation of an upper-bound on the cost functional is needed. This is subsequently exploited to derive a quasi-convex approximation of the robust LQG problem. An inner approximation of the quasi convex problems via FIR truncation is employed. Previous results (\cite{mania2019}, \cite{furieri2020}) show  that indeed the {\em certainty equivalent controller} may achieve superior sub-optimality scaling than our result, but only  for the fully observed LQR settings (\cite{mania2019}), in the setup of a  stricter requirement on admissible uncertainty.  Given the lack of prior gain margin  for the optimal LQG controller, which is known to be notoriously fragile, even under small model uncertainty the stabilizability of the resulted controller may be lost, thus the availability of a more general framework for modelling of uncertainty is important.

Existing non-asymptotic identification methods (\cite{Dahleh2020}) have been adapted in order to yield a comparable end-to-end sample complexity. The identification of the unstable plant is performed in closed loop, directly on its the coprime factors via the dual Youla Parameterization (\cite{Anderson1998}). The  algorithm employed for system identification doesn't require the knowledge of the model's order (\cite{Oymak1}), which is the common scenario in many applications.  
Pursuing the identification of the plant a $\mathcal{H}_2$ bound for Hankel matrix estimation with high probability is derived, followed by a $\mathcal{H}_{\infty}$ bound on the uncertainty on the coprime factors, which quantifies the modeling error. The robust controller design is recast as convex optimization for estimated nominal model within a {\em worst case scenario} on the uncertainty. For the output feedback of potentially unstable plants, the resulted  sample complexity result is matched to the same level as that obtained in recent papers (\cite{Boczar2018}, \cite{dean2020}, \cite{furieri2020}), where the robust so-called {\em SLP} or {\em IOP} procedures (\cite{Zheng2020a}) are used for  design. 




The paper is organized as follows: the general setup  is given in \hyperref[generalsetupandpreliminaries]{Section II}. The robust controller synthesis with uncertainty on the coprime factors is included in \hyperref[robustcontrollersynthesis]{Section III}. The sub-optimality guarantees are discussed in \hyperref[endtoendanalysis]{Section IV}. A brief discussion on the closed-loop system identification is provided in \hyperref[clsysid]{Section V} with end-to-end sample complexity results. Conclusions and future possible directions are given in \hyperref[conclusion]{Section V}. All the proofs are postponed to the \hyperref[appendix]{Appendices}, where literature review, mathematical preliminaries and few remarks also have been discussed briefly.

\section{General Setup and Technical Preliminaries}\label{generalsetupandpreliminaries}
The notation used in this paper is fairly common in control systems. 
Upper and lower case boldface letters (e.g. ${\bf z} $ and ${\bf G}$) are used to denote signals and transfer function matrices, and lower and upper case letters (e.g. $z$ and $A$) are used to denote vectors and matrices. The enclosed results are valid for  discrete-time linear systems, therefore $z$ denotes the complex variable associated  with the $\mathbf{Z}$-transform for discrete-time systems. 
A LTI system is {\em stable} if all the poles of its TFM are situated 
inside the unit circle for discrete time systems. The TFM of a LTI system is called {\em unimodular} if it is square, stable and has a stable inverse. For the sake of brevity the $z$ argument after a transfer function may be omitted. Some frequently used notation is listed in the next page.

 \begin {table}[h!]
\begin{tabular}{ >{\centering\arraybackslash}m{0.8in}  >{\arraybackslash}m{5in} }
\toprule[1.25pt] 
& {\bf Nomenclature of  Basic Notation}  \\
\midrule
 LTI & Linear and Time Invariant \\
 TFM & Transfer Function Matrix \\ 
 DCF & Doubly Coprime Factorization \\
 LCF & Left Coprime Factorization\\
 RCF &  Right Coprime Factorization\\
$x\overset{def}{=} y$ &  $x$ is by definition equal to $y$ \\ 
 $\mathbb{R}(z) $ & Set  of all real--rational  transfer functions \\
 $\mathbb{R}(z)^{p \times m} $ & Set of $p \times m$ matrices having all entries in $\mathbb{R} (z)$ \\

${\bf T}^{ \ell \varepsilon}$ & The TFM of the (closed-loop) map having $\varepsilon$ as  input   and $\ell$  as output \\
${\bf T}^{ \ell \varepsilon}_{\bf Q}$ & The TFM of the (closed-loop) map  from the exogenous signal $\varepsilon$ to the  signal $\ell$ inside the feedback loop, as a function of the Youla parameter ${\bf Q}$\\
$\| {\bf G} \|_{{F}}$  & Frobenius norm, Schur norm or $l_2$ norm of ${\bf G} \in \mathbb{R}(z)$, defined as $|tr(\bf GG^*)|^{1/2}$\\
$\|{\bf G}\|_{\infty}$ &  $\mathcal{H}_{\infty}$-norm of ${\bf G} \in \mathbb{R}(z)$, defined as 
$\sup_{\omega} \sigma_{\max} \big({\bf G}(e^{j \omega})\big)$\\
 $\| {\bf G} \|_{\mathcal{H}_{2}}$  & $\mathcal{H}_{2}$-norm of ${\bf G} \in \mathbb{R}(z)$, defined as 
 $ {{{\dfrac{1}{2\pi} \Big( \bigintsss \limits_{-\pi}^{\pi} tr \big({\bf G}^* (e^{j\omega}){\bf G}(e^{j\omega})\big) d\omega }}\Big)}^{1/2}$\\
$\texttt{pt}$ & Notations for true plant (e.g. ${\bf G}^{\mathtt{pt}}$, ${\bf K}^{\mathtt{pt}}$) \\
$\texttt{md}$ & Notations for nominal/estimated model (e.g. ${\bf G}^{\mathtt{md}}$, ${\bf K}^{\mathtt{md}}$)\\
${\bf G}^{\mathtt{pt}}$, ${\bf K}^{\mathtt{opt}}$ & True Plant, Optimal Controller\\
${\bf G}^{\mathtt{md}}$, ${\bf K}^{\mathtt{md}}_{\bf Q}$ & Estimated Model, Nominal stabilizing controller for any stable Youla parameter ${\bf Q}$\\

\bottomrule[1.25pt]
\end {tabular}
\end {table}

\subsection{Standard Unity Feedback} 


A standard unity feedback configuration is depicted in Figure~\ref{2Block}, where ${\bf G}\in \mathbb{R}(z)^{p \times m}$ is a multi-variable LTI plant and ${\bf K}\in \mathbb{R}(z)^{m \times p}$ is an LTI controller. Here $w$, $\nu$ and $r$  are the input disturbance, sensor noise and reference signal respectively while $u$, $z$ and $y$ are the controls, regulated signals  and measurements vectors, respectively.  
If all the closed--loop maps  from the exogenous signals $\displaystyle [ r^T\; \: w ^T \; \: \nu^T \;]^T$ to 
any point inside the feedback loop of Figure~\ref{2Block} are stable, then ${\bf K}$ is said to be an (internally) stabilizing controller of ${\bf G}$ or equivalently  that ${\bf K}$  stabilizes ${\bf G}$. 
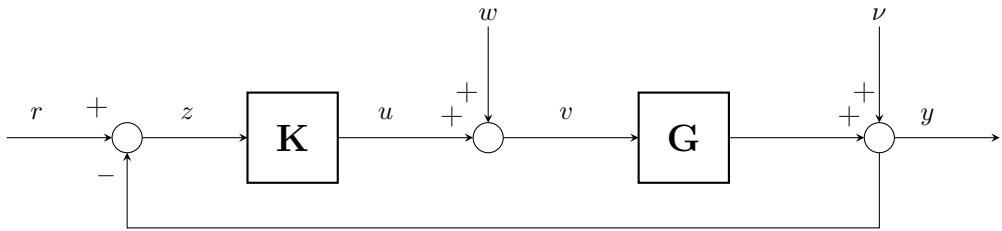
\begin{figure}[h]
\begin{tikzpicture}[scale=0.4]
\draw[xshift=0.1cm, >=stealth ] [->] (0,0) -- (3.5,0);
\draw[ xshift=0.1cm ]  (4,0) circle(0.5);
\draw[xshift=0.1cm] (3,1)   node {\bf{+}} (1,0.8) node {$r$};
\draw [xshift=0.1cm](6,0.8)   node {$z$} ;
\draw[ xshift=0.1cm,  >=stealth] [->] (4.5,0) -- (8,0);
\draw[ thick, xshift=0.1cm]  (8,-1.5) rectangle +(3,3);
\draw [xshift=0.1cm](9.5,0)   node {\Large{${\bf K}$}} ;
\draw[ xshift=0.1cm,  >=stealth] [->] (11,0) -- (15.5,0);
\draw[ xshift=0.1cm ]  (16,0) circle(0.5cm);
\draw [xshift=0.1cm](12.6,0.8)   node {$u$} ;
\draw [xshift=0.1cm](18.6,0.8)   node {$v$} ;
\draw [xshift=0.1cm] (14.8,0.7)   node {\bf{+}};
\draw[  xshift=0.1cm,  >=stealth] [->] (16,3.7) -- (16,0.5);
\draw [xshift=0.1cm] (16,3.6)  node[anchor=south] {$w$}  (15.3,1.5)  node {\bf{+}};
\draw[  xshift=0.1cm,  >=stealth] [->] (16.5,0) -- (21,0);
\draw[ thick, xshift=0.1cm ]  (21,-1.5) rectangle +(3,3) ;
\draw [xshift=0.1cm] (22.5,0)   node {\Large{${\bf G}$}} ;
\draw[ xshift=0.1cm,  >=stealth] [->] (24,0) -- (28.5,0);
\draw[ xshift=0.1cm ] (29,0)  circle(0.5);
\draw [xshift=0.1cm] (29,3.6)  node[anchor=south] {$\nu$}  (28.5,1.5)  node {\bf{+}};
\draw [xshift=0.1cm] (28,0.7)   node {\bf{+}};
\draw [xshift=0.1cm] (30.6,0.7)   node {$y$};
\draw[  xshift=0.1cm,  >=stealth] [->] (29.5,0) -- (33,0);
\draw[  xshift=0.1cm,  >=stealth] [->] (29,3.7) -- (29,0.5);
\draw[ xshift=0.1cm,  >=stealth] [->] (29,-0.5) -- (29,-3) -- (4,-3)-- (4, -0.5);
\draw [xshift=0.1cm] (3.3,-1.3)   node {\bf{--}};
\useasboundingbox (0,0.1);
\end{tikzpicture}
\caption{Standard unity feedback loop of the plant $\bf G$ with the controller $\bf K$}
\label{2Block}
\end{figure}

\noindent
The notation  ${\bf T}^{ \ell \varepsilon}$ is used to indicate  the
 mapping  from signal $\varepsilon$ to signal $\ell$  after combining 
all the ways in which $\ell$ is a function of $\varepsilon$ and solving any
feedback loops that may exist. For example, ${\bf T}^{zw}$ in Figure~\ref{2Block} is the mapping from the disturbances $w$ to the regulated measurements $z$.

\subsection{The Youla-Ku\c{c}era Parameterization}\label{2}

\begin{prop} \label{14Martie2019} 
Given a TFM ${\bf K} \in \mathbb{R}(z)^{m \times p}$, a fractional representation of the form ${\bf K}={\bf R}^{-1}{\bf S}$ with  ${\bf R}\in \mathbb{R}(z)^{m \times m}$, ${\bf S} \in \mathbb{R}(z)^{m \times p}$ is called a {\em left factorization} of {\bf K}. If ${\bf K}={\bf Y}^{-1}{\bf X}$ is a left factorization of ${\bf K}$ then  any other left factorization of ${\bf K}$ such as  ${\bf K}={\bf R}^{-1}{\bf S}$ is of the form ${\bf R}={\bf UY}$, ${\bf S}={\bf UX}$, for some invertible TFM ${\bf U}$.
\end{prop}

\noindent
Given a plant ${\bf G} \in \mathbb{R}(z)^{p \times m}$, a left  coprime factorization of ${\bf G}$  is defined by ${\bf G} = \tilde{\bf M}^{-1}\tilde{\bf N}$,  with $\tilde{\bf N}\in \mathbb{R}(z)^{p \times m}$, $ \tilde {\bf M} \in \mathbb{R}(z)^{p \times p} $ both stable and satisfying $\tilde{\bf M} \tilde{\bf Y} + \tilde{\bf N}\tilde{\bf X} =I_p$, for certain stable TFMs $\tilde{ \bf X} \in \mathbb{R}(z)^{m \times p}$, $\tilde {\bf Y} \in \mathbb{R}(z)^{p \times p}$. \\

\noindent
Analogously, a  right coprime factorization of ${\bf G}$  is defined by ${\bf G} = {\bf NM}^{-1}$ with  both factors ${\bf N}\in \mathbb{R}(z)^{p \times m}$, ${\bf M} \in \mathbb{R}(z)^{m \times m}$ being stable and for which there exist ${\bf X}\in \mathbb{R}(z)^{m \times p}$, ${\bf Y} \in \mathbb{R}(z)^{m \times m} $ also stable, satisfying  ${\bf YM} + {\bf XN} = I_m$  \cite[Ch.~4, Corollary~17]{vidyasagar1985}, with $I_m$ being the identity matrix.\\  


\begin{defn} \cite[Ch.4, Remark pp. 79]{vidyasagar1985} 
A collection of eight  stable TFMs $\big({\bf M}, {\bf  N}$, $\tilde {\bf  M}, \tilde {\bf  N}$, ${\bf X}, {\bf  Y}$, $\tilde {\bf  X}, \tilde {\bf  Y}\big)$ is called a  {\em doubly coprime factorization}  of ${\bf  G}$  if   $\tilde {\bf  M}$ and ${\bf M}$ are invertible, yield the
factorizations
${\bf G}=\tilde{\bf M}^{-1}\tilde{\bf N}={\bf NM}^{-1}$, and satisfy the following equality  (B\'{e}zout's identity):
\begin{equation}\label{bezoutidentity}
\ba{cc}   \tilde {\bf M} &  \tilde {\bf N} \\ - {\bf X} &  {\bf Y} \ea
\ba{cc}  \tilde {\bf Y} & -{\bf N} \\   \tilde {\bf X} &  {\bf M} \ea = I_{p+m},
\ba{cc}  \tilde {\bf Y} & -{\bf N} \\   \tilde {\bf X} &  {\bf M} \ea 
\ba{cc}   \tilde {\bf M} &  \tilde {\bf N} \\ - {\bf X} &  {\bf Y} \ea
 =  I_{p+m}.
\end{equation}
\end{defn}

\vspace{3pt}
\begin{thm} \label{Youlaaa}
{\bf (Youla-Ku\u{c}era)} \cite[Ch.5, theorem 1]{vidyasagar1985} Let  $\big({\bf M}, {\bf  N}$, $\tilde {\bf  M}, \tilde {\bf  N}$, ${\bf X}, {\bf  Y}$, $\tilde {\bf  X}, \tilde {\bf  Y}\big)$ be a doubly coprime factorization of ${\bf G}$. Any controller ${\bf K}_{\bf Q}$ stabilizing the plant ${\bf G}$, in the feedback interconnection of Figure~\ref{2Block}, can be written as
\begin{equation}
\label{YoulaEq}
{\bf K_Q}={\bf Y}_{\bf Q}^{-1}{\bf X_Q} = {\bf \tilde{X}}_{\bf Q}  {\bf \tilde{Y}}_{\bf Q}^{-1},
 \end{equation}
where
${\bf X_Q}$, ${\bf \tilde{X}_Q}$, ${\bf Y_Q}$ and ${\bf \tilde{Y}_Q}$ are defined as:
\begin{equation} \label{5}
\begin{split}
{\bf X_Q}  \overset{def}{=}  {\bf X}+{\bf Q} \tilde{\bf M}, \quad
&{\bf \tilde{X}_Q}  \overset{def}{=}  \tilde{\bf X}+{\bf MQ}, \quad \\
{\bf Y_Q}  \overset{def}{=}  {\bf Y} - {\bf Q} \tilde{\bf N}, \quad 
&{\bf \tilde{Y}_Q}  \overset{def}{=}  \tilde{\bf Y}-{\bf NQ}  
\end{split}
\end{equation}
for some stable ${\bf Q}$ in $\mathbb{R}(z)^{m\times p}$. It also holds that ${\bf K_Q}$ from (\ref{YoulaEq}) stabilizes ${\bf G}$, for any stable ${\bf Q}$ in $\mathbb{R}(z)^{m\times p}$.
\end{thm}

\begin{prop} \label{2018April29}
\label{rem:232pm27feb2012}
Starting from any doubly coprime factorization (\ref{bezoutidentity}), the following identity  
\begin{equation}\label{dcrelQ}
\! \ba{cl} \! {\bf U}_1 \tilde {\bf M} &  {\bf U}_1 \tilde {\bf N} \\  - {\bf U}_2  {\bf X_Q}&   {\bf U}_2 {\bf Y_Q} \! \ea \!
\ba{lr} \!   \tilde {\bf Y}_{\bf Q} {\bf U}^{-1}_1 & -{\bf N} {\bf U}^{-1}_2  \\   \tilde {\bf X}_{\bf Q} {\bf  U}^{-1}_1  & {\bf M} {\bf U}^{-1}_2 \!  \ea = I_{p+m}.
\end{equation}
provides the class of {\em all} doubly coprime factorizations of ${\bf G}$, where ${\bf Q}$ is  stable  in $\mathbb{R}(z)^{m\times p}$ and ${\bf U}_1 \in \mathbb{R}(z)^{p\times p}$, ${\bf U}_2 \in \mathbb{R}(z)^{m \times m}$ are both unimodular.  
\end{prop}

\begin{defn} \label{klap} $\mathbb{H}({\bf G}, {\bf K}_{{\bf Q}})$ denotes the TFM whose entries are the closed-loop maps from $\displaystyle [ r^T\; \: w ^T \; \: \nu^T \; ]^T$ to $\displaystyle [y^T\;\; u^T \;\; z^T \;\; v^T ]^T$ achievable via the stabilizing controllers (\ref{YoulaEq}), namely equation (\ref{klap1}) here
\begin{equation*} \label{klap3}
\ba{c} y\\u\\z\\v\\ \ea = \mathbb{H}({\bf G}, {\bf K}_{{\bf Q}}) \ba{c} r\\w\\ \nu\\ \ea,
\end{equation*}
{\begin{equation} \label{klap1}
\mathbb{H}({\bf G}, {\bf K}_{{\bf Q}})\overset{def}{=}\ba{ccc}  (I_p+{\bf GK_Q})^{-1}{\bf GK_Q} & (I_p+{\bf GK_Q})^{-1}{\bf G} &(I_p+{\bf GK_Q})^{-1}  \\ 
(I_m+{\bf K_Q}{\bf G})^{-1}{\bf K_Q} & -(I_m+{\bf K_QG})^{-1}{\bf K_QG} &-(I_m+{\bf K_Q}{\bf G})^{-1}{\bf K_Q} \\
(I_p+{\bf GK_Q})^{-1} & -(I_p+{\bf GK_Q})^{-1}{\bf G} & -(I_p+{\bf GK_Q})^{-1} \\
(I_m+{\bf K_QG})^{-1} {\bf K_Q}& (I_m+{\bf K_QG})^{-1} & - (I_m+{\bf K_QG})^{-1}{\bf K_Q} \\
\ea 
\end{equation}
}
\end{defn}


\begin{prop} \label{afinecloop} \cite[(7)/~pp.101]{vidyasagar1985} ${\bf T}^{ \ell \varepsilon}_{\bf Q}$ denotes the dependency on the Youla parameter ${\bf Q}$ of the closed loop map from the exogenous signal $\varepsilon$ to the signal $\ell$ inside the feedback loop. The set of {\em all} closed loop maps (\ref{klap1}) achievable via stabilizing controllers (\ref{YoulaEq}) depends on the plant ${\bf G}$  alone and not on the particular doubly coprime factorization (\ref{bezoutidentity})  in which the Youla parameterization is formulated. 
Furthermore, the parameterization of the closed loop maps (\ref{klap1}) is affine in the Youla parameter $\bf{Q}$.
\begin{equation} \label{finallly2}
\small
\begin{tabular}{|c |  c|  c|  c|} %
\hline
${} $ & ${r}$ & ${w}$ & ${v}$ \\
\hline
& & & \\[-2.35ex]
${y} $  & $I_p- \tilde {\bf Y}_{\bf Q} \tilde {\bf M}$ & $\tilde {\bf Y}_{\bf Q} \tilde {\bf N}$  & $\tilde {\bf Y}_{\bf Q} \tilde {\bf M}$ \\ [0.75ex]
\hline
& & & \\[-2.35ex]
${ u} $ & $\tilde {\bf X}_{\bf Q}  \tilde {\bf M}$ & $-\tilde {\bf X}_{\bf Q}  \tilde {\bf N}$ & $-\tilde {\bf X}_{\bf Q}  \tilde {\bf M}$  \\[0.75ex]
\hline
& & & \\[-2.35ex]
${z} $  & $ \tilde {\bf Y}_{\bf Q} \tilde {\bf M}$ & $- \tilde {\bf Y}_{\bf Q} \tilde {\bf N}$ &$- \tilde {\bf Y}_{\bf Q} \tilde {\bf M}$ \\[0.75ex]
\hline
& & & \\[-2.35ex]
${v} $  & $\tilde {\bf X}_{\bf Q}  \tilde {\bf M}$ & $I_m -\tilde {\bf X}_{\bf Q}  \tilde {\bf N}$ & $ {-\tilde {\bf X}_{\bf Q}  \tilde {\bf M}}$ \\[0.75ex]
\hline
\end{tabular}\textcolor{white}{,} 
\begin{tabular}{|c |  c|  c|  c|} %
\hline
${} $ & ${r}$ & ${w}$ & ${v}$ \\
\hline
& & & \\[-2ex]
${y} $  & ${{\bf N}{\bf X}_{\bf Q}}$ & $ {{\bf N}{\bf Y}_{\bf Q}}  $  & $I_p- {\bf N}{\bf X}_{\bf Q}  $ \\ [0.75ex]
\hline
& & & \\[-2ex]
${ u} $ & $ {{\bf M} {\bf X}_{\bf Q}}   $ & $ {- (I_m - {\bf M}{\bf Y}_{\bf Q})}   $ & $- {\bf M} {\bf X}_{\bf Q}   $  \\[0.75ex]
\hline
& & & \\[-2ex]
${z} $  & $  {I_p - {\bf N}{\bf X}_{\bf Q}}  $ & $-  {\bf N}{\bf Y}_{\bf Q}  $ &$-(  I_p - {\bf N}{\bf X}_{\bf Q} ) $ \\[0.75ex]
\hline
& & & \\[-2ex]
${v} $  & $ {\bf M}{\bf X}_{\bf Q}   $ & $ {{\bf M}{\bf Y}_{\bf Q}} $ & $ {- {\bf M}{\bf X}_{\bf Q} }  $ \\[0.75ex]
\hline
\end{tabular}
\end{equation}
\end{prop}

\subsection{Dual Youla-Ku\c{c}era Parameterization}

\begin{thm} \label{DualYoulaaa}
{\bf (Dual Youla-Ku\c{c}era)}(\cite{VanHof1992})  Let  $\big({\bf M}, {\bf  N}$, $\tilde {\bf  M}, \tilde {\bf  N}$, ${\bf X}, {\bf  Y}$, $\tilde {\bf  X}, \tilde {\bf  Y}\big)$ be a doubly coprime factorization of ${\bf G}$. Any plant ${\bf G}_R$ stabilized by a fixed controller ${\bf K}$, can be written as
\begin{equation}
\label{YoulaEqR}
{\bf G_R} = {\bf \tilde{M}_R}^{-1}{\bf \tilde{N}_R} = {\bf N_R}{\bf M_R}^{-1},
 \end{equation}
where
${\bf M_R}$, ${\bf \tilde{M}_R}$, ${\bf N_R}$ and ${\bf \tilde{N}_R}$ are defined as:
\begin{equation} \label{DualEqYoula4}
\begin{split}
{\bf M_R}  \overset{def}{=}  {\bf M}- \tilde{\bf X}{\bf R}, \quad
&{\bf \tilde{M}_R}  \overset{def}{=}  \tilde{\bf M}-{\bf RX}, \quad \\
{\bf N_R}  \overset{def}{=}  {\bf N} +  \tilde{\bf Y}{\bf R}, \quad 
&{\bf \tilde{N}_R}  \overset{def}{=}  \tilde{\bf N}+{\bf RY}  
\end{split}
\end{equation}
for some stable ${\bf R}$ in $\mathbb{R}(z)^{p\times m}$. 
\end{thm}

\begin{prop}\label{dualyoula2}(\cite{VanHof1992})
  Let 
  $\bf{G}$ with LCF 
  ${\bf \tilde{M}}^{-1}{\bf \tilde{N}}$ be any system that is stabilized by a controller $\bf{K}$ with LCF ${\bf Y}^{-1}{\bf X}$.
Then, the plant $\bf{G_R}$ is stabilized by controller $\bf{K}$ iff there exists a stable ${\bf R} \in \mathbb{R}(z)^{p\times m}$, such that $\bf{G_R}$ = $(\tilde{\bf M}-{\bf RX})^{-1}(\tilde{\bf N}+{\bf RY})$. Similarly, for RCF of controller ${\bf K}$ = ${\bf \tilde{X}}{\bf \tilde{Y}}^{-1}$ and plant $\bf{G}$ = ${\bf N}{\bf M}^{-1}$, the plant $\bf{G}$ is stabilized by controller $\bf{K}$ iff there exists a stable ${\bf R} \in \mathbb{R}(z)^{p\times m}$, such that $\bf{G_R}$ = $( {\bf N} +  \tilde{\bf Y}{\bf R}) ({\bf M}- \tilde{\bf X}{\bf R})^{-1}$.
\end{prop}


\section{Robust Controller Synthesis}\label{robustcontrollersynthesis}

Given a DCF of the  nominal model of the plant ${\bf G}^{\mathtt{md}} = (\tilde{\bf M}^{\mathtt{md}})^{-1} \tilde{\bf N}^{\mathtt{md}} = {\bf N}^{\mathtt{md}} ({\bf M}^{\mathtt{md}})^{-1}$, we can write the Bezout identity that incorporates the corresponding Youla parameterization of all stabilizing controller  for the nominal model ${{\bf K}^{\mathtt{md}}_{\bf Q}} = ({\bf Y}^{\mathtt{md}}_{\bf Q})^{-1} {\bf X}^{\mathtt{md}}_{\bf Q} = {\bf {\tilde{X}}}^{\mathtt{md}}_{\bf Q}  ({\bf \tilde{Y}}^{\mathtt{md}}_{\bf Q})^{-1}$ as:

\begin{equation}\label{bezout2}
\ba{cc}   {\tilde {\bf M}}^{\mathtt{md}} &  {\tilde {\bf N}}^{\mathtt{md}} \\ - {{\bf X}}^{\mathtt{md}}_{\bf Q} &  {{\bf Y}}^{\mathtt{md}}_{\bf Q} \ea
\ba{cc}  {\tilde {\bf Y}}^{\mathtt{md}}_{\bf Q} & -{{\bf N}}^{\mathtt{md}} \\   {\tilde {\bf X}}^{\mathtt{md}}_{\bf Q} &  {{\bf M}}^{\mathtt{md}} \ea = \ba{cc}   { I_p} &   {0} \\  {0} &  { I_m} \ea,
\end{equation} 
where {\bf Q} denotes as usually the Youla parameter.

\begin{defn}[Model Uncertainty Set] \label{modeluncertaintyset}
The $\gamma$-radius {\em model uncertainty} set (for the nominal plant ${\bf G}^{\mathtt{md}}$ \textrm{with} \;
$\Delta_{\bf \tilde{M}}$, $\Delta_{\bf \tilde{N}}$ {both stable}) is defined as:
\begin{equation} \label{Youlamd}
    \mathcal{G}_\gamma \overset{\mathrm{def}}{=} \{ {\bf G} = {\tilde {\bf M}}^{-1} {\tilde {\bf N}} \hspace{2pt} \big | \hspace{2pt} {\tilde {\bf M}} = ({\tilde {\bf M}^{\mathtt{md}}}+\Delta_{\bf {\tilde{M}}}),  {\tilde {\bf N}} = ({\tilde {\bf N}}^{\mathtt{md}}+\Delta_{\bf \tilde{N}});  \; \;  \; \Big\| \ba{cc}  \Delta_{\bf \tilde{M}} &  \Delta_{\bf \tilde{N}} \ea \Big\|_{\infty}  < \gamma \}
\end{equation}
\end{defn}

\begin{defn}[$\gamma$-Robustly Stabilizing]\label{perturbedplant} A stabilizing controller ${\bf K}^{\mathtt{md}}$ of the nominal plant  
  is said to be $\gamma$-robustly stabilizing iff ${\bf K}^{\mathtt{md}}$ stabilizes not only ${\bf G}^{\mathtt{md}}$ but also all plants  ${\bf G} \in \mathcal{G}_\gamma$. 
\end{defn}

\begin{assumption}
 It is assumed that the  true plant, denoted by ${\bf G}^{\mathtt{pt}}$, belongs to the   model uncertainty set introduced in \defref{modeluncertaintyset}, i.e. that there exist stable $\Delta_{\bf \tilde{M}}$, $\Delta_{\bf \tilde{N}}$ with  $ \Big\| \ba{cc}  \Delta_{\bf \tilde{M}} &  \Delta_{\bf \tilde{N}} \ea \Big\|_{\infty}  < \gamma$ for which ${\bf G}^{\mathtt{pt}} = ({\tilde {\bf M}}^{\mathtt{md}}+\Delta_{\bf \tilde{M}})^{-1}  ({\tilde {\bf N}}^{\mathtt{md}}+\Delta_{\bf \tilde{N}})$.
\end{assumption}

\noindent In the presence of additive uncertainty on the coprime factors the Bezout identity in \eqref{bezout2} no longer holds, however, the following holds for {\em certain stable}  ${\Delta_{\bf M}}$, ${\Delta_{\bf N}}$ {\em factors}:
\begin{equation}\label{PhiMatrix}
\ba{cc}    ({\tilde {\bf M}}^{\mathtt{md}}+\Delta_{\bf \tilde{M}})  &   ({\tilde {\bf N}}^{\mathtt{md}}+\Delta_{\bf \tilde{N}}) \\ - {{\bf X}}^{\mathtt{md}}_{\bf Q}  &  {{\bf Y}}^{\mathtt{md}}_{\bf Q}  \ea
\ba{cc}  {\tilde {\bf Y}}^{\mathtt{md}}_{\bf Q} & -({{\bf N}}^{\mathtt{md}}+\Delta_{\bf {N}}) \\   {\tilde {\bf X}}^{\mathtt{md}}_{\bf Q} &  ({{\bf M}}^{\mathtt{md}}+\Delta_{\bf {M}}) \ea = \ba{cc}    {\bf \Phi}_{11} &   O \\  O &  {\bf \Phi}_{22} \ea.
\end{equation} 
The block diagonal structure of the right hand side term in \eqref{PhiMatrix} is due to the fact that ${\bf G}^{\mathtt{pt}} = ({\tilde {\bf M}}^{\mathtt{md}}+\Delta_{\bf \tilde{M}})^{-1}  ({\tilde {\bf N}}^{\mathtt{md}}+\Delta_{\bf \tilde{N}}) =({{\bf N}}^{\mathtt{md}}+\Delta_{\bf {N}}) ({{\bf M}}^{\mathtt{md}}+\Delta_{\bf {M}})^{-1} $ for the aforementioned {\em certain stable}  ${\Delta_{\bf M}}$, ${\Delta_{\bf N}}$ {\em factors}.

\begin{lem}\label{phi11phi22}
 A stabilizing controller of the nominal plant ${{\bf K}^{\mathtt{md}}_{\bf Q}} = ({\bf Y}^{\mathtt{md}}_{\bf Q})^{-1} {\bf X}^{\mathtt{md}}_{\bf Q} = {\bf {\tilde{X}}}^{\mathtt{md}}_{\bf Q}  ({\bf \tilde{Y}}^{\mathtt{md}}_{\bf Q})^{-1}$ is  $\gamma$-robustly stabilizing iff  for any stable model perturbations  $\Delta_{\bf \tilde{M}}, \Delta_{\bf \tilde{N}}$ with  $ \Big\| \ba{@{}cc@{}}  \Delta_{\bf \tilde{M}} &  \Delta_{\bf \tilde{N}} \ea \Big\|_{\infty}  < \gamma$ the TFM 
  \begin{equation} \label{August10_7am}
         {\bf \Phi}_{11} = { I_p} + \ba{cc}  \Delta_{\bf \tilde{M}} &  \Delta_{\bf \tilde{N}} \ea \ba{c}  {\tilde {\bf Y}}^{\mathtt{md}}_{\bf Q} \\  {\tilde {\bf X}}^{\mathtt{md}}_{\bf Q} \ea,
 \end{equation}
  is unimodular ({\em i.e.}  it is  square, stable and has an inverse  ${\bf \Phi}_{11}^{-1}$ which is also stable) from \eqref{PhiMatrix}. A similar condition for $\gamma$-robust stabilizability can be formulated in terms of ${\bf \Phi}_{22}$ TFM, whereas 
 \begin{equation}
 {\bf \Phi}_{22} = { I_m} + \ba{cc}  {\bf X}^{\mathtt{md}}_{\bf Q} &  {\bf Y}^{\mathtt{md}}_{\bf Q} \ea \ba{c} \Delta_{\bf {N}} \\  \Delta_{\bf {M}} \ea.
 \end{equation}
 
\end{lem}
Since $  {\bf \Phi}_{11}$  in \eqref{August10_7am} clearly depends on the Youla parameter (via the right coprime factors of the controller), the condition for the $\gamma$-robust stabilizability  of the controller can be recast in the following particular form, which will be instrumental in the sequel:

\begin{thm}\label{oarasthm1} The Youla parameterization yields a $\gamma$-robustly stabilizing controller iff its corresponding Youla parameter satisfies 
$\Bigg\| \ba{c} {\tilde {\bf Y}}^{\mathtt{md}}_{\bf Q} \\  {\tilde {\bf X}}^{\mathtt{md}}_{\bf Q} \ea \Bigg\|_{\infty}  \leq \dfrac{1}{\gamma}$.
\end{thm}

\noindent
The proofs for \lemref{phi11phi22} and \thmref{oarasthm1} are given on \hyperref[appendixC]{Appendix C}.
As an intermediary result, by employing \thmref{oarasthm1} and the standard inequality \eqref{submultiplicativeinfty} from \hyperref[appendixB]{Appendix B} it is concluded that:
\begin{center}
$\Bigg\| \ba{cc}  -\Delta_{\bf \tilde{M}} &  -\Delta_{\bf \tilde{N}} \ea
\ba{c} {\tilde {\bf Y}}^{\mathtt{md}}_{\bf Q} \\  {\tilde {\bf X}}^{\mathtt{md}}_{\bf Q} \ea \Bigg\|_{\infty}\leq \Bigg\| \ba{cc}  \Delta_{\bf \tilde{M}} &  \Delta_{\bf \tilde{N}} \ea \Bigg\|_{\infty} \Bigg\|
\ba{c}  {\tilde {\bf Y}}^{\mathtt{md}}_{\bf Q} \\  {\tilde {\bf X}}^{\mathtt{md}}_{\bf Q} \ea \Bigg\|_{\infty} < \gamma \times \dfrac{1}{
\gamma} = 1$.
\end{center}

\noindent

\noindent



\begin{prop}\label{closedloopresponses11}
The square root of the LQG cost function from \eqref{LQGCostOriginal} has the form:
\begin{equation}\label{cost2}
\begin{split}
    \mathcal{H}({\bf G}^{\mathtt{pt}}, {{\bf K}^{\mathtt{md}}_{\bf Q}}) \overset{def}{=} & \Bigg\|  \ba{c}    {\tilde {\bf Y}}^{\mathtt{md}}_{\bf Q} \\  {\tilde {\bf X}}^{\mathtt{md}}_{\bf Q} \ea {\bf \Phi}_{11}^{-1} \ba{cc}  ({\tilde {\bf M}^{\mathtt{md}}}+\Delta_{\bf \tilde{M}}) &  ({\tilde {\bf N}^{\mathtt{md}}}+\Delta_{\bf \tilde{N}})  \ea \Bigg\|_{\mathcal{H}_{2}}\\
\end{split}
\end{equation}
representing the following closed loop responses:
\begin{equation}\label{closedloopresponses1}
    \ba{cc} {{\bf T}^{y \nu}_{\bf Q}} & {{\bf T}^{yw}_{\bf Q}} \\ {{\bf T}^{u \nu}_{\bf Q}} & {{\bf T}^{u w}_{\bf Q}} \ea 
    = \ba{cc} {(I_p+{\bf G}^{\mathtt{pt}} {{\bf K}^{\mathtt{md}}_{\bf Q}})^{-1}} & { {(I_p+{\bf G}^{\mathtt{pt}} {{\bf K}^{\mathtt{md}}_{\bf Q}})^{-1}}{{\bf G}^{\mathtt{pt}}}} \\
{{{\bf K}^{\mathtt{md}}_{\bf Q}}{(I_p+{\bf G}^{\mathtt{pt}} {{\bf K}^{\mathtt{md}}_{\bf Q}})^{-1}}} &  {{{\bf K}^{\mathtt{md}}_{\bf Q}}{(I_p+{\bf G}^{\mathtt{pt}} {{\bf K}^{\mathtt{md}}_{\bf Q}})^{-1}}}{{\bf G}^{\mathtt{pt}}} \ea
\end{equation}
\end{prop}
\noindent
The proofs for the \propref{closedloopresponses11} and \thmref{theoremAcost} are given on \hyperref[appendixC]{Appendix C}.

\begin{thm}\label{theoremAcost}
The Robust LQG Control Problem is defined as :
\begin{equation} \label{theoremA}
\begin{aligned}
\min_{\bf Q\: \text{stable}}  \quad & \max_{\Big\| \ba{cc}  \Delta_{\bf \tilde{M}} &    \Delta_{\bf \tilde{N}}  \ea \Big\|_\infty < \gamma } \mathcal{H}({\bf G}^{\mathtt{pt}}, {{\bf K}^{\mathtt{md}}_{\bf Q}})   \\[5pt]
\textrm{s.t.} \quad  
\quad &  \Bigg\|
\ba{c}  {\tilde {\bf Y}}^{\mathtt{md}}_{\bf Q} \\  {\tilde {\bf X}}^{\mathtt{md}}_{\bf Q} \ea \Bigg\|_{\infty} \leq \dfrac{1}{
\gamma}.
\end{aligned}
\end{equation}
whose solution, obtained for the optimal Youla parameter ${\bf Q_*}$ in \eqref{theoremA} will be denoted by ${\tilde {\bf Y}}^{\mathtt{md}}_{\bf Q_*}$, ${\tilde {\bf X}}^{\mathtt{md}}_{\bf Q_*}$ such that the optimal, robust controller reads ${{\bf K}^{\mathtt{md}}_{\bf Q_*}} = {\bf {\tilde{X}}}^{\mathtt{md}}_{\bf Q_*}  ({\bf \tilde{Y}}^{\mathtt{md}}_{\bf Q_*})^{-1}$.
\end{thm}

\noindent
Canonical min-max formulation caused by the additive uncertainty on the coprime factors of the robust LQG controller synthesis renders the problem non-convex.
In order to circumvent this, an upper bound on the $\mathcal{H}({\bf G}^{\mathtt{pt}}, {{\bf K}^{\mathtt{md}}_{\bf Q}})$ cost functional will be derived, much in the spirit of \cite{Dean2018} and \cite{furieri2020}. This bound will further  be exploited to derive a quasi-convex approximation for the robust LQG control problem in the next subsection.

\subsection{Quasi-convex formulation}

\begin{prop}\label{LQGcostupperbound2}

Given any $\gamma$-robustly stabilizing controller satisfying $\Bigg\|
\ba{c} {\tilde {\bf Y}}^{\mathtt{md}}_{\bf Q} \\  {\tilde {\bf X}}^{\mathtt{md}}_{\bf Q} \ea \Bigg\|_{\infty} < \dfrac{1}{
\gamma}$ then for any additive model perturbations $\Big\| \ba{cc}  \Delta_{\bf \tilde{M}} &    \Delta_{\bf \tilde{N}}  \ea \Big\|_\infty < \gamma$, 
the cost functional of the robust LQG problem from \eqref{theoremA} admits the upper bound: 

\begin{equation}  \label{LQGcostupperbound3}
  \mathcal{H}({\bf G}^{\mathtt{pt}}, {{\bf K}^{\mathtt{md}}_{\bf Q}}) \leq  \dfrac{1}{1 - \gamma \bigg\| \ba{c}  {\tilde {\bf Y}}^{\mathtt{md}}_{\bf Q} \\  {\tilde {\bf X}}^{\mathtt{md}}_{\bf Q}  \ea \bigg\|_{\infty}}   {\Bigg[ {h \Big( \gamma, \bigg\| \ba{c}   {\tilde {\bf Y}}^{\mathtt{md}}_{\bf Q} \\  {\tilde {\bf X}}^{\mathtt{md}}_{\bf Q} \ea \bigg\|_{\infty} \Big)} \Bigg\| \ba{c}  {\tilde {\bf Y}}^{\mathtt{md}}_{\bf Q} \\  {\tilde {\bf X}}^{\mathtt{md}}_{\bf Q} \ea \bigg\|_{\mathcal{H}_{2}} \Bigg] }
\end{equation}

\noindent
where ${h \bigg( \gamma, \bigg\| \ba{c}   {\tilde {\bf Y}}^{\mathtt{md}}_{\bf Q} \\  {\tilde {\bf X}}^{\mathtt{md}}_{\bf Q} \ea \bigg\|_{\infty} \bigg)} \overset{def}{=} \bigg(1 + {\gamma \bigg\| \ba{c}   {\tilde {\bf Y}}^{\mathtt{md}}_{\bf Q} \\  {\tilde {\bf X}}^{\mathtt{md}}_{\bf Q} \ea \bigg\|_{\infty}  }\bigg)  \bigg( \bigg\| \ba{cc}  {\tilde {\bf M}^{\mathtt{md}}} &  {\tilde {\bf N}^{\mathtt{md}}}  \ea \bigg\|_{\infty} + \gamma \bigg)$.



\end{prop}

\noindent
 Next one of the main results is stated that provides an aptly designed approximation of the robust LQG control problem from Definition~\ref{theoremAcost} by means of the LQG cost upper bound from \eqref{LQGcostupperbound3}. The proofs for \propref{LQGcostupperbound2}  and \thmref{theoremB} are given in \hyperref[appendixE]{Appendix E}.

\begin{thm}\label{theoremB}
For the true plant, ${\bf G}^{\mathtt{pt}} \in \mathcal{G}_\gamma$ 
and $(\forall) \alpha >0$, the robust LQG control problem in \eqref{theoremA} admits the following upper bound:
\begin{equation} \label{theoremB2}
\begin{aligned}
        \min_{\delta \in [0,1/\gamma)} & \dfrac{1}{1 - \gamma \delta}  \min_{\bf Q\: \text{stable}}  {\bigg( {h \big(\gamma, \alpha\big)} \Bigg\| \ba{c}  {\tilde {\bf Y}}^{\mathtt{md}}_{\bf Q} \\  {\tilde {\bf X}}^{\mathtt{md}}_{\bf Q} \ea \bigg\|_{\mathcal{H}_{2}} \bigg)}\\ \vspace{5pt}
       \textrm{s.t.} & \Bigg\| \ba{c}  {\tilde {\bf Y}}^{\mathtt{md}}_{\bf Q} \\  {\tilde {\bf X}}^{\mathtt{md}}_{\bf Q}\ea \Bigg\|_{\infty}  \leq \min \{\delta, \alpha\}, 
\end{aligned}
\end{equation}


\noindent where as before, ${h \big( \gamma, \alpha \big)}$ = $ (1 + \gamma \alpha) \bigg( \bigg\| \ba{cc}  {\tilde {\bf M}^{\mathtt{md}}} & {\tilde {\bf N}^{\mathtt{md}}}  \ea \bigg\|_{\infty} + \gamma\bigg) $.\\
\end{thm}

\begin{rem} \label{numerical} (Numerical Computation) For each fixed $\delta$, the inner optimization is convex but the dimension of ${\bf Q}(z)$ remains infinite.For numerical computation, a  (FIR) truncation is considered on ${\bf Q}(z)$, such that an equivalent Semi-Definite Program (SDP) can be formulated for the inner optimization problem. Details of this is discussed in \hyperref[appendixE]{Appendix E}.

\end{rem}

\begin{rem} (Feasibility)
Since the subsequent quasi-convex optimization problem in Theorem-\ref{theoremB} must include an additional constraint $\bigg\| \ba{c}   {\tilde {\bf Y}}^{\mathtt{md}}_{\bf Q} \\  {\tilde {\bf X}}^{\mathtt{md}}_{\bf Q} \ea \bigg\|_{\infty} < \alpha$ proportional to $\alpha$, this $\alpha$ should be chosen as small as possible.
However,  $\alpha$ can't be made arbitrarily small and therefore the   feasibility of the quasi-convex program from Theorem-\ref{theoremB} cannot be guaranteed. This is caused by the fact that  $\alpha$ must be assimilated to the norm of a LCF of a stabilizing controller for {\em the true plant},  whose $H_\infty$ attenuation is simultaneously greater or equal to the one of the optimal $H_\infty$ controller {\em for the nominal model} (see Theorem~\ref{oarasthm1}). This reflects the fact that the feasibility of the closed loop "learning" problem  depends inherently on the performance (with respect to the model uncertainty of the plant) of the initially chosen stabilizing controller (the one with which  the closed loop "learning" is being performed).

\end{rem}

\section{Analysis of End-to-End Performance}\label{endtoendanalysis}
\subsection{Sub-optimality guarantee}
\noindent
If we denote by ${{\bf K}^{\mathtt{opt}}}$  the optimal $\mathcal{H}_2$ controller for the true plant, then by Assumption~1 there exist stable additive factors such that ${\bf G}^{\mathtt{pt}}$ = $(\tilde{\bf M}^{\mathtt{md}} + \Delta_{\bf \tilde{M}})^{-1} (\tilde{\bf N}^{\mathtt{md}} + \Delta_{\bf \tilde{N}})$ = $({\bf N}^{\mathtt{md}}+\Delta_{\bf {N}})({\bf M}^{\mathtt{md}}+\Delta_{\bf {M}})^{-1}$ and furthermore, there always exists a Bezout identity of the true plant  that  features the optimal controller ${{\bf K}^{\mathtt{opt}}} = ({\bf Y}^{\mathtt{opt}})^{-1} {\bf X}^{\mathtt{opt}} = {\bf {\tilde{X}}}^{\mathtt{opt}}  ({\bf \tilde{Y}}^{\mathtt{opt}})^{-1}$ as its ``central controller'', thus reading:
\begin{equation}\label{RealPlantOptimalController1}
\ba{cc}   ({\tilde {\bf M}}^{\mathtt{md}}+\Delta_{\bf \tilde{M}}) &   ({\tilde {\bf N}}^{\mathtt{md}}+\Delta_{\bf \tilde{N}}) \\ - {{\bf X}}^{\mathtt{opt}} &  {{\bf Y}}^{\mathtt{opt}}  \ea
\ba{cc}  {\tilde {\bf Y}}^{\mathtt{opt}} & -({{\bf N}}^{\mathtt{pt}}+\Delta_{\bf {N}}) \\   {\tilde {\bf X}}^{\mathtt{opt}} &  ({{\bf M}}^{\mathtt{pt}}+\Delta_{\bf {M}}) \ea = \ba{cc}    { I_p} &    {{0}} \\   {0} &   {I_m} \ea
\end{equation} 
\noindent
\noindent
Consequently, the square root of the LQG cost functional for optimal controller is given by:
\begin{equation}\label{cost3}
    \mathcal{H}({\bf G}^{\mathtt{pt}}, {\bf K}^{\mathtt{opt}}) \overset{def}{=}   \Bigg\| \ba{c}  {\tilde {\bf Y}}^{\mathtt{opt}} \\  {\tilde {\bf X}}^{\mathtt{opt}} \ea \ba{cc}   ({\tilde {\bf M}}^{\mathtt{md}}+\Delta_{\bf \tilde{M}}) &   ({\tilde {\bf N}}^{\mathtt{md}}+\Delta_{\bf \tilde{N}})    \ea \Bigg\|_{\mathcal{H}_{2}}
\end{equation}

\noindent
Next the main result on the sub-optimality guarantee for the performance of the robust controller with model uncertainty of radius $\gamma$ is stated. The proof for \thmref{suboptimalitygurantee} is given in \hyperref[appendixF]{Appendix F}. 


\begin{thm}\label{suboptimalitygurantee}
Let ${{\bf K}^{\mathtt{opt}}}$ be the optimal LQG controller and ${{\bf G}^{\mathtt{pt}}}$ be the model of the true plant, with modeling error uncertainty satisfying ${\Big\| \ba{cc}  \Delta_{\bf \tilde{M}} &    \Delta_{\bf \tilde{N}}  \ea \Big\|_\infty < \gamma }$. Furthermore, let ${\bf Q_*}$ and $\delta_*$  denote the solution to \eqref{theoremB2}. 
Then, when applying the resulting controller ${{\bf K}^{\mathtt{md}}_{\bf Q_*}}$ in feedback interconnection with the true plant ${{\bf G}^{\mathtt{pt}}}$, the relative error in the LQG cost is upper bounded by:

\begin{equation}\label{relativeerror}
    \dfrac{ \mathcal{H}({\bf G}^{\mathtt{pt}}, {{\bf K}^{\mathtt{md}}_{\bf Q_*}})^2 -  \mathcal{H}({\bf G}^{\mathtt{pt}}, {\bf K}^{\mathtt{opt}})^2 }{ \mathcal{H}({\bf G}^{\mathtt{pt}}, {\bf K}^{\mathtt{opt}})^2 } \leq  \dfrac{1}{\Big(1-\gamma \bigg\| \ba{c} {\bf \tilde{Y}}^{\mathtt{md}}_{\bf Q_*} \\  {\bf \tilde{X}}^{\mathtt{md}}_{\bf Q_*} \ea \bigg\|_{\infty} \Big)^2} \times {g \bigg( \gamma, \bigg\| \ba{c}   {\tilde {\bf Y}}^{\mathtt{md}}_{\bf Q_*} \\  {\tilde {\bf X}}^{\mathtt{md}}_{\bf Q_*} \ea \bigg\|_{\infty} \bigg)}^2 -1, 
\end{equation}
\begin{equation*}
    \textrm{where} \hspace{4pt}
    \small {g \bigg( \gamma, \bigg\| \ba{c}   {\tilde {\bf Y}}^{\mathtt{md}}_{\bf Q_*} \\  {\tilde {\bf X}}^{\mathtt{md}}_{\bf Q_*} \ea \bigg\|_{\infty} \bigg)} \overset{def}{=} \bigg\| \ba{c}  {\tilde {\bf Y}}^{\mathtt{md}}_{\bf Q_*} \\  {\tilde {\bf X}}^{\mathtt{md}}_{\bf Q_*} \ea \bigg\|_{\infty} \bigg(1 + {\gamma \bigg\| \ba{c}   {\tilde {\bf Y}}^{\mathtt{md}}_{\bf Q_*} \\  {\tilde {\bf X}}^{\mathtt{md}}_{\bf Q_*} \ea \bigg\|_{\infty}  }\bigg)  \bigg( \bigg\| \ba{cc}  {\tilde {\bf M}^{\mathtt{md}}} &  {\tilde {\bf N}^{\mathtt{md}}}  \ea \bigg\|_{\infty} + \gamma \bigg) .
\end{equation*}
\end{thm}



\noindent
\begin{rem} (Optimality vs Robustness) If $\gamma \bigg\| \ba{c} {\bf \tilde{Y}}^{\mathtt{md}}_{\bf Q_*} \\  {\bf \tilde{X}}^{\mathtt{md}}_{\bf Q_*} \ea \bigg\|_{\infty} = \eta$, it's easy to observe that $\eta \in (0,1)$.  Then it's immediate to see that the upper bound of the relative error in the LQG cost increases as a function of $\eta$. The price of obtaining a faster rate 
is that the controller becomes less robust to model uncertainty as pointed out in \cite{mania2019}, \cite{furieri2020}. It holds for this case too as shown in \thmref{suboptimalitygurantee}. In practice, using a relatively large value for $\eta$ forces a trade-off of optimality for robustness in the controller design procedure. In general, optimality stands i.e. better controller performance is guaranteed as $\eta$ goes closer to $0$ and better robustness performance is guaranteed as $\eta$ goes closer to $1$ with the upper bound \eqref{relativeerror} of relative error in LQG cost might be large. This is shown with an example below.

\noindent
Let's set $\eta = \dfrac{1}{5}$.
 Then by \thmref{suboptimalitygurantee} relative error in the LQG cost is
 \begin{equation}
    \dfrac{ \mathcal{H}({\bf G}^{\mathtt{pt}}, {{\bf K}^{\mathtt{md}}_{\bf Q_*}})^2 -  \mathcal{H}({\bf G}^{\mathtt{pt}}, {\bf K}^{\mathtt{opt}})^2 }{ \mathcal{H}({\bf G}^{\mathtt{pt}}, {\bf K}^{\mathtt{opt}})^2 } \leq  2 \times {g \bigg( \gamma, \bigg\| \ba{c}   {\tilde {\bf Y}}^{\mathtt{md}}_{\bf Q_*} \\  {\tilde {\bf X}}^{\mathtt{md}}_{\bf Q_*} \ea \bigg\|_{\infty} \bigg)} \bigg]^2 -1.
\end{equation}
Hence,  the the relative error in the LQG cost grows as $\mathcal{O}(\gamma^2)$ as long as $\gamma \bigg\| \ba{c} {\bf \tilde{Y}}^{\mathtt{md}}_{\bf Q_*} \\  {\bf \tilde{X}}^{\mathtt{md}}_{\bf Q_*} \ea \bigg\|_{\infty} < \dfrac{1}{5}$.
\end{rem}

\section{Closed Loop Identification Scheme} \label{clsysid}

\noindent

Figure~2 at the top of next page depicts the closed-loop identification setup of a potentially unstable {\em noise contaminated plant} $\textbf{G}^\mathtt{md}$ with control input $u$, noise $\nu$ (taken $w$ = $0$) and output measurement $y$ (where  $u$ and  $\nu$ are assumed independent and stationary), provided that some initial stabilizing controller ${{\bf K}^{\mathtt{md}}}$ is available beforehand. 
\begin{figure}[ht]
\centering
\includegraphics[width=15.6cm]{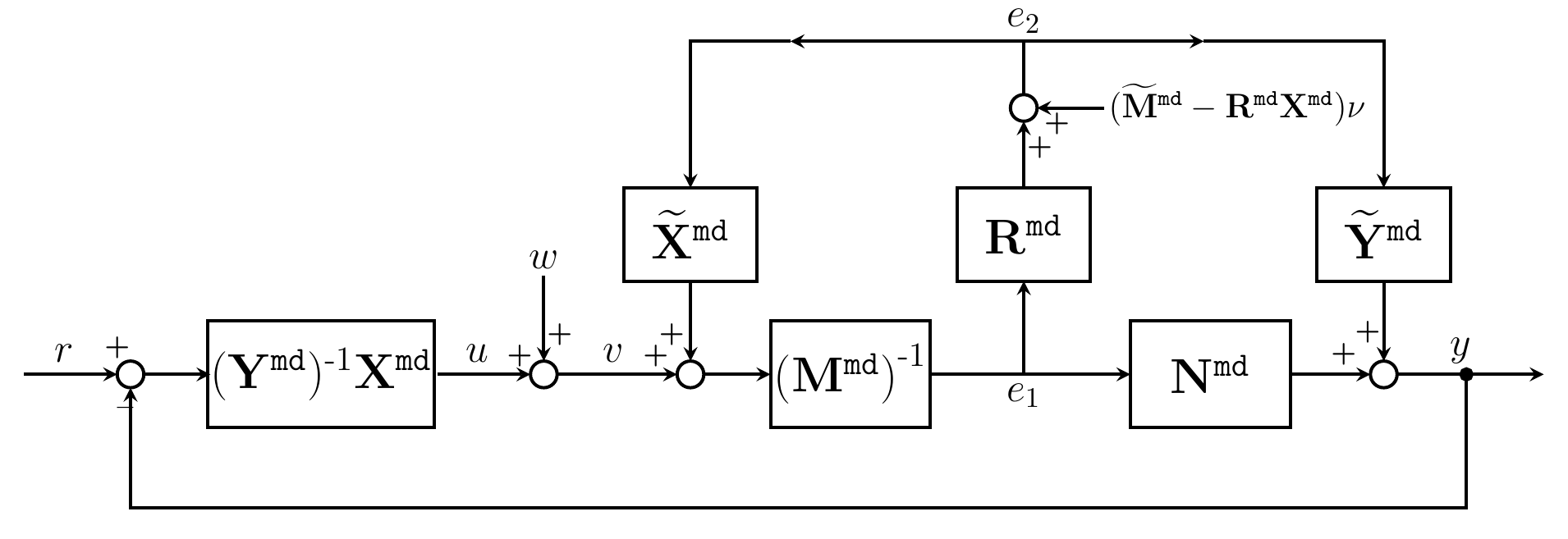}
\caption{Closed loop Identification for Noise Contaminated Plant}
\label{clsysid2}
\end{figure}

\noindent
If any plant $\mathbf{G}^\mathtt{pt}$ is LTI and stable, for open-loop identification it can be written that:
\begin{equation}
    y = \textbf{G}^\mathtt{pt} u + \nu 
\end{equation}
\noindent
Here $u$ and $y$ are available through measurements, then the model of the plant can be estimated as $\mathbf{G}^\mathtt{md}$  (\cite{Anderson1998}). 
\noindent
Suppose that the controller ${{\bf K}^{\mathtt{md}}} = ({\bf Y}^{\mathtt{md}})^{-1} {\bf X}^{\mathtt{md}} = {\bf {\tilde{X}}}^{\mathtt{md}}  ({\bf \tilde{Y}}^{\mathtt{md}})^{-1}$, where ${\bf X}^{\mathtt{md}}$, ${\bf Y}^{\mathtt{md}}$ are left coprimes and ${\bf \tilde{X}}^{\mathtt{md}}$, ${\bf \tilde{Y}}^{\mathtt{md}}$ are right coprimes of ${{\bf K}^{\mathtt{md}}}$. Then by the coprimeness there exists $ {{\bf M}}^{\mathtt{md}}$, $ { {\bf N}}^{\mathtt{md}}$, $ {\tilde {\bf M}}^{\mathtt{md}}$, $ {\tilde {\bf N}}^{\mathtt{md}}$ with
\begin{equation}\label{coprime1}
   {\bf {X}}^{\mathtt{md}}{\bf {N}}^{\mathtt{md}} + {\bf {Y}}^{\mathtt{md}}{\bf {M}}^{\mathtt{md}} = {I}_m, \hspace{2pt} {\bf \tilde{M}}^{\mathtt{md}}{\bf \tilde{Y}}^{\mathtt{md}} + {\bf \tilde{N}}^{\mathtt{md}}{\bf \tilde{X}}^{\mathtt{md}} = {I}_p.
\end{equation}
\noindent
This means that ${\bf G}^{\mathtt{md}}$  = ${\bf {N}}^{\mathtt{md}} {({\bf {M}}^{\mathtt{md}})}^{-1}$ = ${({\bf \tilde{M}}^{\mathtt{md}})}^{-1}{\bf \tilde{N}}^{\mathtt{md}}$ is the nominal model stabilized by ${\bf K}^{\mathtt{md}}$. Alternatively, it can be started with a RCF representation ${\bf {N}}^{\mathtt{md}} ({\bf {M}}^{\mathtt{md}})^{-1}$ of the nominal model ${\bf G}^{\mathtt{md}}$ and then choose that particular representation ${\bf {\tilde{X}}}^{\mathtt{md}}  ({\bf \tilde{Y}}^{\mathtt{md}})^{-1}$ of the known stabilizing ${\bf K}^{\mathtt{md}}$ so that equation \eqref{coprime1} holds. The set of all plants ${\bf G}_{\bf R}^{\mathtt{md}}$ stabilized by ${\bf K}^{\mathtt{md}}$ is given by Dual Youla-Ku\c{c}era Parameterization in \propref{dualyoula2}:
\begin{equation}\label{coprime2}
    {\bf G}_{\bf R}^{\mathtt{md}} = {({\bf {N}}^{\mathtt{md}}+{\bf {\tilde{Y}}}^{\mathtt{md}}\mathbf{R}^\mathtt{md})({\bf {M}}^{\mathtt{md}} -{\bf {\tilde{X}}}^{\mathtt{md}}\mathbf{R}^\mathtt{md})^{-1}}
\end{equation}
for some stable $\mathbf{R}^\mathtt{md} \in \mathbb{R}^{p \times m}$ and ${\bf {\tilde{X}}}^{\mathtt{md}}  ({\bf \tilde{Y}}^{\mathtt{md}})^{-1}$ is the right coprime factorization of ${\bf K}^{\mathtt{md}}$.

\noindent
Suppose that ${({\bf \tilde{M}}^{\mathtt{md}})}^{-1}{\bf \tilde{N}}^{\mathtt{md}}$ is a left coprime fractional description of the nominal model ${\bf G}^{\mathtt{md}}$ such that
\begin{equation}
     \begin{bmatrix} {\mathbf{Y}}^{\mathtt{md}} & {\mathbf X}^{\mathtt{md}}\\ -{\mathbf {\tilde{N}}}^{\mathtt{md}} & {\mathbf {\tilde{M}}}^{\mathtt{md}}  \end{bmatrix} 
\begin{bmatrix}  {\mathbf{M}}^{\mathtt{md}} & - {\mathbf{\tilde{X}}}^{\mathtt{md}}  \\  {\mathbf {N}}^{\mathtt{md}} & {\mathbf {\tilde{Y}}}^{\mathtt{md}}   \end{bmatrix} = \begin{bmatrix} {I}_m  & 0 \\ 0 & {I}_p  \end{bmatrix} 
\end{equation}

\noindent
\begin{lem}\label{noisecontaminatedplant}
The noise contaminated plant set-up of \figref{clsysid2} is identical to \figref{2Block} with input $u$, noise $\nu$ and output $y$, it implies
\begin{equation}\label{GwithR}
    y = {\bf G}_{\bf R}^{\mathtt{md}} u + \nu 
\end{equation}
\noindent
where ${\bf G}_{\bf R}^{\mathtt{md}}$ is given by \eqref{coprime2}.
\end{lem}

\noindent
The key idea dating back to \cite{Anderson1998} is to identify the {\em stable} dual-Youla parameter ${\bf R}^{\mathtt{md}}$ from Theorem~\ref{DualYoulaaa} rather than ${\bf G}^{\mathtt{md}}$, thus recasting the problem in a standard, open-loop identification form. More specifically, direct inspection of Figure~2 shows that
\begin{equation}\label{zblock2}
e_2 = {\bf R}^{\mathtt{md}} e_1 + ({\bf \tilde{M}}^{\mathtt{md}} - {\bf R}^{\mathtt{md}}{\bf X}^{\mathtt{md}})\nu = {\bf R}^{\mathtt{md}}(e_1 - {\bf X}^{\mathtt{md}} \nu )  + {\bf \tilde{M}}^{\mathtt{md}}\nu.
\end{equation}
In \eqref{zblock2} we have the knowledge of ${\bf X}^{\mathtt{md}}$; $e_1$ and $e_2$ are  available from measurements with noise $\nu$. The recent algorithm from \cite{Dahleh2020} given below is employed toward identifying the dual-Youla parameter.\\
Details of this idea and proof for \lemref{noisecontaminatedplant} is given in \hyperref[appendixG1]{Appendix G.1}, also choice of noise is discussed in \hyperref[appendixG1]{Appendix G.2}.

\subsection{Identification Algorithms} 
\noindent
Identification algorithms from \cite{Dahleh2020} has been employed here. Preliminaries of system identification is given in \hyperref[appendixG]{Appendix G.3}.
The state space representation of ${\bf R}^{\mathtt{md}}$ in $e_2 = {\bf R}^{\mathtt{md}} u  + {\bf \tilde{M}}^{\mathtt{md}}\nu$ with $u$ = $e_1 - {\bf X}^{\mathtt{md}} \nu$ is:
\begin{equation}\label{system}
\begin{aligned}
 l_{t+1} &= A_R l_t + B_R u_t + \eta_{t+1}\\
    z_t &= C_R l_t + {\bf \tilde{M}}^{\mathtt{md}} \nu_t
\end{aligned}
\end{equation}
\noindent
Here two important assumptions is stated as below:
\begin{assumption} \label{noiseassumption}
 The noise process $\{\eta_{t}\}_{t=1}^{\infty}$ in the dynamics of ${\bf R}^{\mathtt{md}}$ are i.i.d., and isotropic with sub-gaussian parameter 1. 
The noise process $\{r_{t}\}_{t=1}^{\infty}$, $\{w_{t}\}_{t=1}^{\infty}$ and $\{\nu_{t}\}_{t=1}^{\infty}$ are Gaussian processes with mean function $m_r(t) = m_w(t) = m_{\nu}(t)$ $0$, and their spectral density ${\bf \phi}_r(\omega)$, ${\bf \phi}_w(\omega)$ and ${\bf \phi}_{\nu}(\omega)$ satisfy the constraint in \hyperref[appendixG]{Appendix G}.
\end{assumption}

\begin{assumption} \label{existenceofR}
 There exists constants $\beta, R \geq 1$ s.t. 
$\|\mathcal{T}_{0,\infty}\|_2 < \beta$ and $\dfrac{\|\mathcal{TO}_{k,d}\|_2}{\|\mathcal{T}_{0,\infty}\|_2} \leq \mathcal{R}$.\\
 $\beta$ exists since ${\bf R}^{\mathtt{md}}$ is stable.
\end{assumption}

\noindent
The Algorithms 
for system identification is illustrated in \hyperref[appendixG]{Appendix G}.
\noindent
\subsubsection{Probabilistic guarantees} \label{probabilisticgurantee}




\noindent
Let's define, $T_*(\delta)$ = $\inf \{ T | d_*(T, \delta) \in \mathcal{D}(T), d_*(T,\delta) \leq 2 d_*(\frac{T}{256}, \delta) \}$
where, $d_*(T,\delta)$ = $\inf \{ d|16  \beta \mathcal{R} \alpha(d) $ $ \geq \Big\| \mathcal{\hat{H}}_{0,\hat{d},\hat{d}} - \mathcal{\hat{H}}_{0,\infty,\infty}\Big\|_2 \}$,
with $\mathcal{\hat{H}}_{p,q,r}$ is the $(p, q, r)$ - dimensional estimated Hankel matrix.
Whenever $T \geq T_*(\delta)$ for the failure probability $\delta$, then it follows with the probability at least $(1-\delta)$ that
\begin{equation}\label{probabilitybound}
    \Big\| \mathcal{\hat{H}}_{0,\hat{d},\hat{d}} - \mathcal{\hat{H}}_{0,\infty,\infty}\Big\|_2  \leq 12c\beta \mathcal{R} \bigg(\sqrt{\dfrac{m\hat{d}+p\hat{d}^2+ \hat{d}log(T/\delta)}{T}}\bigg) .
\end{equation}

 
 
 
 
 
 
  





\subsection{Sample Complexity}
\begin{lem}\label{HRbound} The norm of the identification error incurred by the proposed scheme after Algorithm-\ref{identificationalgorithm} and \ref{identificationalgorithm2} is bounded by: 

$\Big\| {\bf{R}}^\mathtt{md} - {\bf R}^\mathtt{pt} \Big\|_{\infty} \leq  \Big\| \mathcal{\hat{H}}_{0,\hat{d},\hat{d}} - \mathcal{\hat{H}}_{0,\infty,\infty}\Big\|_2 \leq  12c\beta \mathcal{R} \bigg(\sqrt{\dfrac{m\hat{d}+p\hat{d}^2+ \hat{d}log(T/\delta)}{T}}\bigg) $
\end{lem}


\noindent
Finally,  the error on the model uncertainty can be directly checked using \lemref{HRbound} as:\\
\noindent
\begin{equation*}
\small
\begin{aligned}
\Big\| \ba{cc} \Delta_{\bf \Tilde{M}} & \Delta_{\bf \Tilde{N}} \ea\Big\|_{\infty} & = \Big\| \ba{cc} (\Tilde{\bf M}^\mathtt{md} + {\bf X}^\mathtt{md} {\bf R}^\mathtt{md}) & ({\Tilde{\bf N}}^\mathtt{md} - {\bf Y}^\mathtt{md} {\bf R}^\mathtt{md}) \ea - \ba{cc} (\Tilde{\bf M}^{\mathtt{md}} + {\bf X}^{\mathtt{md}} {\bf R}^{\mathtt{pt}}) & ({\Tilde{\bf N}}^{\mathtt{md}} - {\bf Y}^{\mathtt{md}} {\bf R}^{\mathtt{pt}}) \ea \Big\|_{\infty} \\ & = \Big\| \ba{cc} {\bf X}^{\mathtt{md}} & {\bf Y}^{\mathtt{md}} \ea ({\bf R}^\mathtt{md} - {\bf R}^\mathtt{pt}) \Big\|_{\infty} \leq \Big\|\ba{cc} {\bf X}^\mathtt{md} & {\bf Y}^\mathtt{md} \ea \Big\|_{\infty} \Big\| {\bf R}^\mathtt{md} - {\bf{R}}^\mathtt{pt}\Big\|_{\infty} 
\end{aligned}
\end{equation*}
\noindent

\noindent
Consequently, the uncertainty level on the LCF of the model satisfies $\Big\| \ba{cc} \Delta_{\bf \Tilde{M}} & \Delta_{\bf \Tilde{N}} \ea\Big\|_{\infty} < \gamma $.

\begin{thm}\label{finaltheorem}
Define ${s}$ = $144 \Big\|\ba{cc} {\bf X}^\mathtt{md} & {\bf Y}^\mathtt{md} \ea \Big\|_{\infty}^2 c^2 \beta^2 \mathcal{R}^2$.
Then, the robust controller will achieve the relative cost within the bound with probability $(1 - \delta)$ provided $T \geq \max\{T_s, T_*({\delta}) \}$.
Here, $T_s$ is the right most zero of $g(T)$ = $\gamma^2 T - s\hat{d}log(T/\delta) - s(m\hat{d}+p\hat{d}^2)$. 
\noindent
If $g(T)$ doesn't have any zero for $T>0$, then define $T_s = 0$ and 
$T_*(\delta)$ = $\inf \{ T | d_*(T, \delta) \in \mathcal{D}(T), d_*(T,\delta) \leq 2d_*  (\frac{T}{256}, \delta) \}$,

\noindent
where, $d_*(T,\delta)$ = $\inf \{ d | 16 \beta \mathcal{R} \alpha(d) \geq \Big\| \mathcal{\hat{H}}_{0,d,d} - \mathcal{\hat{H}}_{0,\infty,\infty}\Big\|_2 \}$,

\noindent
$\mathcal{D}(T) = \{ d \in \mathbb{N} | d \leq \dfrac{T}{cm^2 log^3(Tm/\delta)} \} $ and $\alpha(h) = \sqrt{h}.\big(\sqrt{\dfrac{m+hp+log(T/\delta)}{T}}\big)$.
\end{thm}







\noindent
Combining \thmref{finaltheorem} with \thmref{suboptimalitygurantee}, it follows that with high probability the suboptimality gap behaves as
\begin{equation*}
    \dfrac{ \mathcal{H}({\bf G}^{\mathtt{pt}}, {{\bf K}^{\mathtt{md}}_{\bf Q_*}})^2 -  \mathcal{H}({\bf G}^{\mathtt{pt}}, {\bf K}^{\mathtt{opt}})^2 }{ \mathcal{H}({\bf G}^{\mathtt{pt}}, {\bf K}^{\mathtt{opt}})^2 } \sim \mathcal{O}\Big(\sqrt{\dfrac{logT}{T}}\Big)
\end{equation*}
\noindent

Finally, we note here that the resulted sample complexity is on par with the existing methods from \cite{furieri2020} and \cite{Dean2018}.


\section{Conclusion and Future work}\label{conclusion}
\noindent
In this paper, we have provided the sample complexity bounds for a robust controller synthesis procedure for LQG problems with unknown dynamics, able to cope with unstable plants. We combined finite-time, non-parametric LTI system identification with the Youla parameterization for robust stabilization under uncertainty on the coprime factors of the plant. 
One exciting avenue for future research is the online learning  LQG control problem under the same type of model uncertainty. Another direction is to work out the  sample complexity of learning the optimal state feedback (LQR) controller in tandem with  the optimal state-observer (Kalman Filter (\cite{matni2020}))  for a potentially  unstable system. Combining these two results, should yield precisely the optimal LQG controller discussed above and reveal the {\em separation principle} within this framework.




\vskip 0.2in
\bibliography{ref}



\tableofcontents
\addcontentsline{toc}{section}{Appendix}
\renewcommand{\theHsection}{A\arabic{section}}

\section*{Appendix}\label{appendix}
\noindent
This appendix is divided into eight parts. 
 \hyperref[appendixA]{Appendix A} presents a brief review about the related works on LTI systems with coprime Factorization, controller parameterization, non-asymptotic system identification and robust control in reinforcement learning. 
 \hyperref[appendixB]{Appendix B} provides a handful mathematical preliminaries on norm identities and inequalities (\cite{optimalcontrol}). 
 An overview of the classical Youla Parameterization framework is given in  \hyperref[appendixC]{Appendix C}. Also the closed loop maps of the system, proof for \lemref{phi11phi22}, \thmref{oarasthm1} and \propref{closedloopresponses11} is stated here. 
 \hyperref[appendixD]{Appendix D} presents the proof for \thmref{theoremAcost}, which includes the equivalent $\mathcal{H}_2$ norm LQG cost derivation and proof for \propref{LQGcostupperbound2}.  
 \hyperref[appendixE]{Appendix E} discusses about the quasi-convex approximation with proof for \thmref{theoremB}. Also numerical computation has been discussed for Finite Impulse Response (FIR) truncation which leads Semi-Definite Program (SDP) formulation for quasi-convex problem. 
 \hyperref[appendixF]{Appendix F} completes the suboptimality guarantee proof in \thmref{suboptimalitygurantee}. 
 \hyperref[appendixG]{Appendix G} presents the non-asymptotic system identification (\cite{Dahleh2020}) for identifying Dual-Youla parameter by using Ordinary Least Square algorithm. $\mathcal{H}_{\infty}$ bound \lemref{HRbound} is also proved here. \\ 



\appendix

\section{Related Works}\label{appendixA}

 Recently, considerable research efforts have been spent towards the finite time (non-asymptotic) learning of the optimal  \textit{Linear Quadratic Gaussian (LQG)} regulator for an {\em unknown} plant by employing the modern statistical and optimization tools from the machine learning framework. \\

\textit{Identification of Unknown Dynamical Systems}: Model identification of  systems has a long history in electrical engineering. The work of \cite{Ljung1999} covers the classical asymptotic results on identification for Linear and Time Invariant ({\bf LTI}) models. 
Classical methods are based on subspaces identification in the absence of disturbances or noise, whereas state-space parameters of a system are identified with the Hankel matrix obtained from input-output measurements, followed by balanced realization of the Grammians, via Hankel singular value decomposition (\cite{Ljung1999}, \cite{Moor2012}). 
More recent approaches  use rank minimization methods for system identification by relaxing rank constraint as in \cite{Fazel2013}, \cite{Grussler2018}. 
Very recently,  LTI systems identification from a single trajectory, with statistical guarantees (of high probability)  has been proposed  based on a standard ordinary least squares criterion (OLS) in \cite{Simchowitz2018}, \cite{Sarkar2019}, using measure concentration type results. Time series online prediction is presented in this framework  in \cite{hazan2018}, \cite{Agarwal2018}. More recent non-asymptotic system identification algorithm inspired by the Kalman-Ho algorithm from a single input-output trajectory is included in \cite{Oymak1}, \cite{Dahleh2020}.  Multiple independent trajectories based system estimation are presented in \cite{Tu2017}, \cite{Zheng2020}. The algorithm for the system estimation in \cite{Oymak1} needs the prior knowledge of the model order which seriously limits its applicability in practice. Other related approaches such as 
\cite{Campi2002}, \cite{Shah2012}, \cite{Hardt2018} suffer from the same limitation. However, recent results such as the algorithm for system estimation in \cite{Dahleh2020} are able to surpass this problem. 
Their identification approach is to estimate the system Hankel-like matrix from noisy data using OLS with high probability. 
However, the method is restricted to open-loop stable  systems. To overcome this hypothesis, the results  from \cite{VanHof1992} have been used in this paper for  parameterizing all plants/models which are stabilized by a prescribed controller, known as the Dual-Youla parameterization (\cite{vidyasagar1985}), which re-casts the identification of the unstable, noise contaminated plant into the problem of identifying the stable dual Youla parameter. Therefore, the key identification idea is to identify the dual Youla parameter from data rather than identifying the model itself, which is shown as an open loop identification problem. As the dual Youla parameter is stable, the algorithm from \cite{Dahleh2020} can be used to identify it with high probability. For more details about dual Youla parameterizations and converting closed loop approximation problems to open loop we refer to \cite{Schrama1991}, \cite{Anderson1998}.\\

\textit{Controller Design}: There is a rich history of convex parameterizations of stabilizing controllers, under some restricting hypotheses guaranteeing their existence (\cite{convexparameterization}). In general, iterative schemes allowing the  design of a controller in tandem with the system identification of the plant render (highly) non-convex problems and  the direct search for the controller inefficient, even intractable. Despite the non-convexity problem, convergence guarantees have been given for gradient-based algorithms over controller parameterization in few recent works such as \cite{Fazel2018}, \cite{Malik2019}, \cite{furieri2020distributed}. 
Another approach is based on Lyapunov-based parameterizations in the state-space domain to obtain linear matrix inequalities (LMIs) as in  \cite{Boyd1994}, \cite{Scherer1997}. Changing variables in different coordinates is also a popular approach in controller parameterization (\cite{Gahinet1994}, \cite{optimalcontrol}). A parameterization over state-feedback policies has been proposed in \cite{Goulart2006}.
The classical  Youla-Kucera parameterization from \cite{youla1976}, \cite{kucera1975} requires only approximate knowledge of the system and gives a unique way in the frequency domain to get a convex formulation.  
Recently, another frequency based parameterization has been proposed and dubbed System Level Synthesis (SLS) in \cite{matni2019}, \cite{Matni2021}. More recently, Input-Output Parameterization (IOP) has been introduced in \cite{furieri2019}. Ultimately, some sort of equivalence exists between IOP, SLS and Youla, under the notion of closed loop convexity (\cite{Boyd1991}), which is basically entails designing the closed loop maps directly rather than designing a controller. These equivalences between IOP, SLS and Youla has been discussed extensively in \cite{Zheng2020a}, \cite{Tseng2021}. However, and this is important, there are still essential distinctions between these paramterizations in terms of their ability to allow the design and the subsequent distributed implementation of optimal controllers. Here preliminary results show that certain crippling working hypotheses, such as the stability of the plant (assumed in the SLS, IOP frameworks), can be elegantly overcome via the Youla parameterization.\\

\textit{Robust Control}:
Non-asymptotic estimation provides a family of plants described by a nominal model with a set of bounded uncertainty/modeling error for a fixed amount of data. Hence it is necessary to ensure that the designed controller is robustly stable with performance guarantee for the family of plants, known as robust controller design (\cite{optimalcontrol}).  
When the modeling errors is unstructured and arbitrary norm-bounded LTI operators, classical ``small-gain'' type theorems can be used to solve the robust stabilization problem with no conservatism. If the modeling error is structured, other sophisticated techniques similar to $\mu$-synthesis techniques in \cite{Doyle1982}, integral quadratic constraints (IQC) in \cite{Megretski1997}, and sum-of-squares (SOS) optimization in  \cite{Scherer2006} are less traditional than small-gain approaches as in \cite{furieri2020}.
The strongest method available in the theoretical control machinery for  coping with  systems with uncertainty  are additive or multiplicative perturbations  on the coprime factors of the nominal model (\cite{robustcontrol}). The modeling errors are allowed to be unstructured and arbitrary  LTI operators, with a mere upper bound on their $\mathcal{H}_\infty$ norm. This describes simultaneously uncertainty on the locations of both the poles and the Smith zeros of the plant and on its McMillan degree. Modeling uncertainty as additive perturbations on the coprimes also has the advantage over  additive uncertainty on the plant, in the fact that it enables to deal with unstable plants.\\

\textit{Online Control}:
Applying online learning techniques for the learning and optimal control of  LTI systems, via  time-varying cost functions has received a lot of attention within the last year. These results bring significant algorithmic benefits over the more classical online and adaptive control methodologies, which have an abundant history of research (\cite{Sastry2011}, \cite{Ioannou2012}). Here we restrict the literature review to online control with unknown dynamics in the regret minimization framework similar to \cite{Hazan2020b}. The notion of ``regret'' arises from the time varying cost function. Essentially, {\em the regret} measures how well a system adapts to the time varying costs. Most of the previous  results in regret minimization framework assume either Gaussian perturbations (or no perturbations at all) for the fully-observed linear dynamic system and manage to obtaining relatively very low regret margins in the online LQR setting as depicted in \cite{Yadkori2011}, \cite{Dean2018},  \cite{mania2019}, \cite{Cohen2019}. A regret bound for fully observable systems with adversarial disturbances able to cope with partially observable systems (under convexity assumption for adversarial disturbances and strong-convexity assumptions for semi-adversarial noises) in the  time-varying  case has been very recently proposed in \cite{Hazan2020b}. Another regret bound for the case of stochastic perturbations, time-varying strongly-convex functions, and partially observed states has been provided in \cite{Lale2020}. Regret minimization with Kalman filtering has been discussed in \cite{Sabag2021}. Distributed online LQ problem has been studied for identical LTI systems in \cite{Chang2020} and for model-free systems in \cite{furieri2020distributed}.

\section{Norm/Inequality preliminaries}\label{appendixB}
\noindent
Some useful norm identities and inequalities for $\mathcal{H}_2$ and $\mathcal{H}_{\infty}$ is presented here from \cite{optimalcontrol} that are key to the proofs. 

\noindent
First, the following norm triangular inequalities is defined:
\begin{equation}\label{triangularinequality}
\begin{split}
    \| {\bf G_1 +G_2} \|_{\mathcal{H}_{2}} \leq  \| {\bf G_1 } \|_{\mathcal{H}_{2}} + \| {\bf G_2} \|_{\mathcal{H}_{2}}, (\forall) {\bf G_1, G_2} \in \mathcal{RH}_{2} \\
    \| {\bf G_1 +G_2} \|_{\infty} \leq  \| {\bf G_1 } \|_{\infty} + \| {\bf G_2} \|_{\infty},  (\forall) {\bf G_1, G_2} \in \mathcal{RH}_{\infty}
\end{split}    
\end{equation}

\noindent
Since $\|.\|_{\infty}$ is an induced norm, we have the following sub-multiplicative property
\begin{equation}\label{submultiplicativeinfty}
    \| {\bf G_1 G_2} \|_{\infty} \leq  \| {\bf G_1 } \|_{\infty} \| {\bf G_2} \|_{\infty},  (\forall) {\bf G_1, G_2} \in \mathcal{RH}_{\infty}
\end{equation}
\noindent
The sub multiplicative property in \eqref{submultiplicativeinfty} does not hold for the $\mathcal{H}_2$ norm since it is not an induced norm. Instead we have the result below.
\begin{lem}\label{submultiplicativeh2}
For TFM ${\bf G_1}$, ${\bf G_2}$ $\in \mathcal{RH}_{\infty}$, it can be stated that
  \begin{equation}\label{submultiplicativeh2norm}
  \begin{split}
          \| {\bf G_1 G_2} \|_{\mathcal{H}_{2}} \leq  \| {\bf G_1 } \|_{\mathcal{H}_{2}} \| {\bf G_2} \|_{\infty},  (\forall) {\bf G_1} \in \mathcal{RH}_{2}, {\bf G_2} \in \mathcal{RH}_{\infty} \\
           \| {\bf G_1 G_2} \|_{\mathcal{H}_{2}} \leq  \| {\bf G_1 } \|_{\infty} \| {\bf G_2} \|_{\mathcal{H}_{2}},  (\forall) {\bf G_1} \in \mathcal{RH}_{\infty}, {\bf G_2} \in \mathcal{RH}_{2}
  \end{split}
\end{equation}
\end{lem}
\noindent
Finally, from small-gain theorem in \thmref{smallgaintheorem}, we have the following useful inequalities by taking power expansion of the matrix inverse.
\begin{lem}
For TFM ${\bf G_1}$, ${\bf G_2}$ $\in \mathcal{RH}_{\infty}$ with $\| {\bf G_1 } \|_{\infty} \times \| {\bf G_2} \|_{\infty} < 1$, we have
\begin{equation}
    \begin{split}
       \| (I - {\bf G_1 G_2 })^{-1}\|_{\infty} \leq \dfrac{1}{1 - \| {\bf G_1 } \|_{\infty} \| {\bf G_2} \|_{\infty}}\\
    \| (I + {\bf G_1 G_2 })^{-1}\|_{\infty} \leq \dfrac{1}{1 - \| {\bf G_1 } \|_{\infty} \| {\bf G_2} \|_{\infty}}
    \end{split}
\end{equation}
\end{lem}

\section{Closed Loop Maps}\label{appendixC}

\noindent
The Bezout identity is retrieved before finding the closed loop maps associated with it by using \eqref{PhiMatrix} as below:
\begin{equation*}\label{bezout4phi}
 \ba{cc}   { \bf \Phi}_{11}^{-1} &  0 \\ 0 & I_m \ea
\ba{cc}  { \bf \Phi}_{11} & 0 \\   0 & { \bf \Phi}_{22} \ea 
\ba{cc}  I_p & 0 \\  0 &  { \bf \Phi}_{22}^{-1}\ea = \ba{cc}   { I_p} &   {0} \\  {0} &  { I_m} \ea ,
\end{equation*} 

 \begin{equation*}\label{Phitoclosedloopmap}
 \ba{cc}   {\bf \Phi}_{11}^{-1} &  0 \\ 0 & I_m \ea
\ba{cc}    ({\tilde {\bf M}}^{\mathtt{md}}+\Delta_{\bf \tilde{M}}) &   ({\tilde {\bf N}}^{\mathtt{md}}+\Delta_{\bf \tilde{N}}) \\ - {{\bf X}}^{\mathtt{md}}_{\bf Q}  &  {{\bf Y}}^{\mathtt{md}}_{\bf Q}  \ea
\ba{cc}  {\tilde {\bf Y}}^{\mathtt{md}}_{\bf Q} & -({{\bf N}}^{\mathtt{md}}+\Delta_{\bf {N}}) \\   {\tilde {\bf X}}^{\mathtt{md}}_{\bf Q} &  ({{\bf M}}^{\mathtt{md}}+\Delta_{\bf {M}}) \ea
\ba{cc}  I_p & 0 \\  0 & {{\bf \Phi}}_{22}^{-1}\ea = \ba{cc}   { I_p} &   {0} \\  {0} &  { I_m} \ea ,
\end{equation*} 

 \begin{equation}\label{closedloopmap}
\small \ba{cc}    {\bf \Phi}_{11}^{-1} ({\tilde {\bf M}}^{\mathtt{md}}+\Delta_{\bf \tilde{M}}) &   {\bf \Phi}_{11}^{-1} ({\tilde {\bf N}}^{\mathtt{md}}+\Delta_{\bf \tilde{N}}) \\ - {{\bf X}}^{\mathtt{md}}_{\bf Q}  &  {{\bf Y}}^{\mathtt{md}}_{\bf Q}  \ea
\ba{cc}  {\tilde {\bf Y}}^{\mathtt{md}}_{\bf Q} & -({{\bf N}}^{\mathtt{md}}+\Delta_{\bf {N}}) {\bf \Phi}_{22}^{-1} \\   {\tilde {\bf X}}^{\mathtt{md}}_{\bf Q} &  ({{\bf M}}^{\mathtt{md}}+\Delta_{\bf {M}}) {\bf \Phi}_{22}^{-1} \ea
= \ba{cc}   { I_p} &   {0} \\  {0} &  { I_m} \ea.
\end{equation} \\

\begin{proof}\textbf{for \lemref{phi11phi22}}:
From DCF matrix ${\bf \Phi}$ in \eqref{PhiMatrix}, ${\bf \Phi}_{11}$ = $({\tilde {\bf M}^{\mathtt{md}}}+\Delta_{\bf {\tilde{M}}}) {\tilde {\bf Y}}^{\mathtt{md}}_{\bf Q}  + ({\tilde {\bf N}^{\mathtt{md}}}+\Delta_{\bf {\tilde{N}}}) {\tilde {\bf X}}^{\mathtt{md}}_{\bf Q}$ and ${\bf{\Phi}}_{22}$ = ${\bf X}^{\mathtt{md}}_{\bf Q} ({\bf N}^{\mathtt{md}}+\Delta_{\bf N})+{\bf Y}^{\mathtt{md}}_{\bf Q}({\bf M}^{\mathtt{md}}+\Delta_{\bf M}) $.
Next using Bezout identity for nominal model in \eqref{bezout2} it follows that
\begin{equation*}\label{phimtrixproof1}
\begin{aligned}
{\bf \Phi}_{11} &= \ba{cc} ({\tilde {\bf M}^{\mathtt{md}}}+\Delta_{\bf {\tilde{M}}}) &  ({\tilde {\bf N}^{\mathtt{md}}}+\Delta_{\bf {\tilde{N}}})\ea
\ba{c} {\tilde {\bf Y}}^{\mathtt{md}}_{\bf Q} \\  {\tilde {\bf X}}^{\mathtt{md}}_{\bf Q} \ea\\
&= \ba{cc}  {\tilde {\bf M}^{\mathtt{md}}} &  {\tilde {\bf N}^{\mathtt{md}}} \ea
\ba{c} {\tilde {\bf Y}}^{\mathtt{md}}_{\bf Q} \\ {\tilde {\bf X}}^{\mathtt{md}}_{\bf Q} \ea + \ba{cc}  \Delta_{\bf \tilde{M}} &  \Delta_{\bf \tilde{N}} \ea
\ba{c}  {\tilde {\bf Y}}^{\mathtt{md}}_{\bf Q} \\  {\tilde {\bf X}}^{\mathtt{md}}_{\bf Q} \ea = I_p + \ba{cc}  \Delta_{\bf \tilde{M}} &  \Delta_{\bf \tilde{N}} \ea
\ba{c}  {\tilde {\bf Y}}^{\mathtt{md}}_{\bf Q} \\  {\tilde {\bf X}}^{\mathtt{md}}_{\bf Q} \ea ;
\end{aligned}
\end{equation*}

\begin{equation*}\label{phimtrixproof2}
\begin{aligned}
{\bf \Phi}_{22} &= \ba{cc} {\bf X}^{\mathtt{md}}_{\bf Q} &  {\bf Y}^{\mathtt{md}}_{\bf Q} \ea
\ba{c} ({\bf N}^{\mathtt{md}}+\Delta_{\bf N}) \\  ({\bf M}^{\mathtt{md}}+\Delta_{\bf M}) \ea\\
&= \ba{cc}  {\bf X}^{\mathtt{md}}_{\bf Q} &  {\bf Y}^{\mathtt{md}}_{\bf Q} \ea
\ba{c} {\bf N}^{\mathtt{md}} \\ {\bf M}^{\mathtt{md}} \ea + \ba{cc}  {\bf X}^{\mathtt{md}}_{\bf Q} &  {\bf Y}^{\mathtt{md}}_{\bf Q} \ea
\ba{c}  \Delta_{\bf N} \\ \Delta_{\bf M} \ea = I_m +  \ba{cc}  {\bf X}^{\mathtt{md}}_{\bf Q} &  {\bf Y}^{\mathtt{md}}_{\bf Q} \ea \ba{c}  \Delta_{\bf N} \\ \Delta_{\bf M} \ea .
\end{aligned}
\end{equation*}
\end{proof}
\noindent
Before giving the proof for \thmref{oarasthm1}, the small gain theorem is stated here.
\noindent
\begin{thm}[Small Gain Theorem]\label{smallgaintheorem} \cite[Theorem~7.4.1/ page 225 ]{robustcontrol}
Let ${\bf G}_1 \in \mathbb{R}(z)^{p\times m}$ and ${\bf G}_2 \in \mathbb{R}(z)^{m \times p}$ be two TFM's respectively. If $\|{\bf G}_1\|_{\infty} \leq \dfrac{1}{ \gamma}$ and $\|{\bf G}_2\|_{\infty} \leq \gamma$, for some $\gamma > 0$, then the closed loop feedback system of ${\bf G}_1$ and ${\bf G}_2$ is internally stable.
\end{thm}

\begin{proof}\textbf{for \thmref{oarasthm1}}:
For any stable ${\bf Q}$ satisfying $\small \Bigg\| \ba{c}  {\tilde {\bf Y}}^{\mathtt{md}}_{\bf Q} \\  {\tilde {\bf X}}^{\mathtt{md}}_{\bf Q} \ea \Bigg\|_{\infty}  \leq \dfrac{1}{ \gamma}$ it follows that ${\bf \Phi}_{11} = \bigg({I_p} +$ $\small \ba{cc}  \Delta_{\bf \tilde{M}} &  \Delta_{\bf \tilde{N}} \ea \ba{c}  {\tilde {\bf Y}}^{\mathtt{md}}_{\bf Q} \\  {\tilde {\bf X}}^{\mathtt{md}}_{\bf Q} \ea \bigg)$ is unimodular (square and stable with a stable inverse) due to the fact that: {\em (a)} The term  $\bigg(I_p + \small \ba{cc}  \Delta_{\bf \tilde{M}} &  \Delta_{\bf \tilde{N}} \ea \ba{c}  {\tilde {\bf Y}}^{\mathtt{md}}_{\bf Q} \\  {\tilde {\bf X}}^{\mathtt{md}}_{\bf Q} \ea\bigg)$ is stable since all factors are stable and {\em (b)} We know that $\small \Big\| \ba{cc}  \Delta_{\bf \tilde{M}} &  \Delta_{\bf \tilde{N}} \ea \Big\|_{\infty}  < \gamma$ from the definition of the Model Uncertainty Set. At the same time $\Bigg( I_p + \small \ba{cc}  \Delta_{\bf \tilde{M}} &  \Delta_{\bf \tilde{N}} \ea \ba{c}  {\tilde {\bf Y}}^{\mathtt{md}}_{\bf Q} \\  {\tilde {\bf X}}^{\mathtt{md}}_{\bf Q} \ea \Bigg)^{-1}$ is guaranteed to be stable via the Small Gain Theorem.\\


\noindent
Conversely, if a Youla parameter ${\bf Q}$  yields a $\gamma$-robustly stabilizable controller of the nominal model  then necessarily $\small \Bigg\| \ba{c}  {\tilde {\bf Y}}^{\mathtt{md}}_{\bf Q} \\  {\tilde {\bf X}}^{\mathtt{md}}_{\bf Q} \ea \Bigg\|_{\infty} \leq \dfrac{1}{ \gamma}$.
The proof of this claim is done by contradiction. Assume that $\small \Bigg\| \ba{c}  {\tilde {\bf Y}}^{\mathtt{md}}_{\bf Q} \\  {\tilde {\bf X}}^{\mathtt{md}}_{\bf Q} \ea \Bigg\|_{\infty}  > \dfrac{1}{\gamma}$.
Then by the Spectral Mapping Theorem \cite[page~41-42]{douglas1972} there must exist $\small \Big\| \ba{cc}  \Delta_{\bf \tilde{M}} &  \Delta_{\bf \tilde{N}} \ea \Big\|_{\infty}  < \gamma$ such that  ${\bf \Phi}_{11} = \bigg( {I_p} + $ $\small \ba{cc}  \Delta_{\bf \tilde{M}} &  \Delta_{\bf \tilde{N}} \ea \ba{c}  {\tilde {\bf Y}}^{\mathtt{md}}_{\bf Q} \\  {\tilde {\bf X}}^{\mathtt{md}}_{\bf Q} \ea \bigg)$ is not unimodular and consequently the Youla parameter ${\bf Q}$ does not produce an $\gamma$-robustly stabilizable controller, which is a contradiction. The proof ends. 
\end{proof}


\begin{proof}\textbf{of \propref{closedloopresponses11}}: By using \eqref{cost2} the cost function can be calculated as
\begin{equation}\label{cost5}
\begin{split}
     \mathcal{H}({\bf G}^{\mathtt{pt}}, {{\bf K}^{\mathtt{md}}_{\bf Q}}) = & \Bigg\|  \ba{c}     {\tilde {\bf Y}}^{\mathtt{md}}_{\bf Q} \\   {\tilde {\bf X}}^{\mathtt{md}}_{\bf Q} \ea {\bf \Phi}_{11}^{-1} \ba{cc}  ({\tilde {\bf M}}^{\mathtt{md}}+\Delta_{\bf \tilde{M}}) &   ({\tilde {\bf N}}^{\mathtt{md}}+\Delta_{\bf \tilde{N}})  \ea \Bigg\|_{\mathcal{H}_{2}}\\
     = & \Bigg\|  \ba{cc}  {\tilde {\bf Y}}^{\mathtt{md}}_{\bf Q}{\bf \Phi}_{11}^{-1} ({\tilde {\bf M}}^{\mathtt{md}}+\Delta_{\bf \tilde{M}})   &  {\tilde {\bf Y}}^{\mathtt{md}}_{\bf Q}{\bf \Phi}_{11}^{-1} ({\tilde {\bf N}}^{\mathtt{md}}+\Delta_{\bf \tilde{N}}) \\  {\tilde {\bf X}}^{\mathtt{md}}_{\bf Q}{\bf \Phi}_{11}^{-1} ({\tilde {\bf M}}^{\mathtt{md}}+\Delta_{\bf \tilde{M}})   &  {\tilde {\bf X}}^{\mathtt{md}}_{\bf Q}{\bf \Phi}_{11}^{-1} ({\tilde {\bf N}}^{\mathtt{md}}+\Delta_{\bf \tilde{N}}) \ea  \Bigg\|_{\mathcal{H}_{2}}\\
     = & \Bigg\| \ba{cc} {\tilde {\bf Y}}^{\mathtt{md}}_{\bf Q}{\bf \Phi}_{11}^{-1} ({\tilde {\bf M}}^{\mathtt{md}}+\Delta_{\bf \tilde{M}})   &  {\tilde {\bf Y}}^{\mathtt{md}}_{\bf Q}{\bf \Phi}_{11}^{-1} ({\tilde {\bf N}}^{\mathtt{md}}+\Delta_{\bf \tilde{N}}) \\  {\tilde {\bf X}}^{\mathtt{md}}_{\bf Q}{\bf \Phi}_{11}^{-1} ({\tilde {\bf M}}^{\mathtt{md}}+\Delta_{\bf \tilde{M}})  & {\tilde {\bf X}}^{\mathtt{md}}_{\bf Q}{\bf \Phi}_{11}^{-1} ({\tilde {\bf N}}^{\mathtt{md}}+\Delta_{\bf \tilde{N}}) \ea  \Bigg\|_{\mathcal{H}_{2}}
\end{split}
\end{equation}

\noindent
Next the closed loop maps is calculated by using \eqref{bezout2} which leads to:
\begin{equation*}
  {\tilde {\bf Y}}^{\mathtt{md}}_{\bf Q}{\bf \Phi}_{11}^{-1} ({\tilde {\bf M}}^{\mathtt{md}}+\Delta_{\bf \tilde{M}}) 
   = {\tilde {\bf Y}}^{\mathtt{md}}_{\bf Q} {\Bigg({I_p} + \ba{cc}  \Delta_{\bf \tilde{M}} &  \Delta_{\bf \tilde{N}} \ea \ba{c}  {\tilde {\bf Y}}^{\mathtt{md}}_{\bf Q} \\  {\tilde {\bf X}}^{\mathtt{md}}_{\bf Q} \ea \Bigg)}^{-1} ({\tilde {\bf M}}^{\mathtt{md}}+\Delta_{\bf \tilde{M}}) 
\end{equation*}

\begin{equation*}
\begin{split}
  {\tilde {\bf Y}}^{\mathtt{md}}_{\bf Q}{\bf \Phi}_{11}^{-1} ({\tilde {\bf M}}^{\mathtt{md}}+\Delta_{\bf \tilde{M}}) 
  & = \Bigg[({\tilde {\bf M}}^{\mathtt{md}}+\Delta_{\bf \tilde{M}})^{-1}  {\bigg({I_p} + \ba{cc}  \Delta_{\bf \tilde{M}} &  \Delta_{\bf \tilde{N}} \ea \ba{c}  {\tilde {\bf Y}}^{\mathtt{md}}_{\bf Q} \\  {\tilde {\bf X}}^{\mathtt{md}}_{\bf Q} \ea \bigg)} ({\tilde {\bf Y}}^{\mathtt{md}}_{\bf Q})^{-1} \Bigg]^{-1}\\
  & = \Bigg[({\tilde {\bf M}}^{\mathtt{md}}+\Delta_{\bf \tilde{M}})^{-1}  {\bigg({({\tilde {\bf M}}^{\mathtt{md}} {\tilde {\bf Y}}^{\mathtt{md}}_{\bf Q} + {\tilde {\bf N}}^{\mathtt{md}} {\tilde {\bf X}}^{\mathtt{md}}_{\bf Q})} + (  \Delta_{\bf \tilde{M}}{\tilde {\bf Y}}^{\mathtt{md}}_{\bf Q} +  \Delta_{\bf \tilde{N}}{\tilde {\bf X}}^{\mathtt{md}}_{\bf Q}) \bigg)} ({\tilde {\bf Y}}^{\mathtt{md}}_{\bf Q})^{-1} \Bigg]^{-1}\\
   & = \Bigg[({\tilde {\bf M}}^{\mathtt{md}}+\Delta_{\bf \tilde{M}})^{-1}  {\bigg({({\tilde {\bf M}}^{\mathtt{md}}+\Delta_{\bf \tilde{M}}) {\tilde {\bf Y}}^{\mathtt{md}}_{\bf Q}+ ({\tilde {\bf N}}^{\mathtt{md}}+\Delta_{\bf \tilde{N}}) {\tilde {\bf X}}^{\mathtt{md}}_{\bf Q}} \bigg)} ({\tilde {\bf Y}}^{\mathtt{md}}_{\bf Q})^{-1} \Bigg]^{-1}\\
   & = \Bigg[ { {I_p} + ({\tilde {\bf M}}^{\mathtt{md}}+\Delta_{\bf \tilde{M}})^{-1}  ({\tilde {\bf N}}^{\mathtt{md}}+\Delta_{\bf \tilde{N}}) {\tilde {\bf X}}^{\mathtt{md}}_{\bf Q} ({\tilde {\bf Y}}^{\mathtt{md}}_{\bf Q})^{-1}} \Bigg]^{-1} = {(I_p+{\bf G}^{\mathtt{pt}} {{\bf K}^{\mathtt{md}}_{\bf Q}})^{-1}}.\\
\end{split}
\end{equation*}

\noindent
Next the remaining three terms is calculated as below:


\begin{equation*}
    {\tilde {\bf Y}}^{\mathtt{md}}_{\bf Q}{\bf \Phi}_{11}^{-1} ({\tilde {\bf N}}^{\mathtt{md}}+\Delta_{\bf \tilde{N}})  = {\tilde {\bf Y}}^{\mathtt{md}}_{\bf Q}{\bf \Phi}_{11}^{-1} ({\tilde {\bf M}}^{\mathtt{md}}+\Delta_{\bf \tilde{M}}) ({\tilde {\bf M}}^{\mathtt{md}}+\Delta_{\bf \tilde{M}})^{-1}  ({\tilde {\bf N}}^{\mathtt{md}}+\Delta_{\bf \tilde{N}}) = {{(I_p+{\bf G}^{\mathtt{pt}} {{\bf K}^{\mathtt{md}}_{\bf Q}})^{-1}}{\bf G}^{\mathtt{pt}}};
\end{equation*}

\begin{equation*}
    {\tilde {\bf X}}^{\mathtt{md}}_{\bf Q}{\bf \Phi}_{11}^{-1} ({\tilde {\bf M}}^{\mathtt{md}}+\Delta_{\bf \tilde{M}}) = {\tilde {\bf X}}^{\mathtt{md}}_{\bf Q}({\tilde {\bf Y}}^{\mathtt{md}}_{\bf Q})^{-1} {\tilde {\bf Y}}^{\mathtt{md}}_{\bf Q}{\bf \Phi}_{11}^{-1} ({\tilde {\bf M}}^{\mathtt{md}}+\Delta_{\bf \tilde{M}})
   = {{{\bf K}^{\mathtt{md}}_{\bf Q}}{(I_p+{\bf G}^{\mathtt{pt}} {{\bf K}^{\mathtt{md}}_{\bf Q}})^{-1}}};\\ 
\end{equation*}


\begin{equation*}
   {\tilde {\bf X}}^{\mathtt{md}}_{\bf Q}{\bf \Phi}_{11}^{-1} ({\tilde {\bf N}}^{\mathtt{md}}+\Delta_{\bf \tilde{N}}) = {\tilde {\bf X}}^{\mathtt{md}}_{\bf Q}{\bf \Phi}_{11}^{-1} ({\tilde {\bf M}}^{\mathtt{md}}+\Delta_{\bf \tilde{M}})({\tilde {\bf M}}^{\mathtt{md}}+\Delta_{\bf \tilde{M}})^{-1}  ({\tilde {\bf N}}^{\mathtt{md}}+\Delta_{\bf \tilde{N}}) = {{{\bf K}^{\mathtt{md}}_{\bf Q}}{(I_p+{\bf G}^{\mathtt{pt}} {{\bf K}^{\mathtt{md}}_{\bf Q}})^{-1}}{\bf G}^{\mathtt{pt}}}.
\end{equation*}
\end{proof} 
\section{Controller Parameterization}\label{appendixD}

\begin{proof}\textbf{for \thmref{theoremAcost}}:
First $\mathcal{H}_2$ norm of the LQG cost in \eqref{LQGCostOriginal} is derived. Youla Parameterization focuses on the responses of a closed loop system (\cite{youla1976}, \cite{furieri2019}).
By \propref{closedloopresponses11} the closed Loop responses with the LQG cost \eqref{cost2} are:
\useshortskip
\begin{equation}\label{closedloopresponses2}
\useshortskip
    \ba{cc} {{\bf T}^{y \nu}_{\bf Q}} & {{\bf T}^{yw}_{\bf Q}} \\ {{\bf T}^{u \nu}_{\bf Q}} & {{\bf T}^{u w}_{\bf Q}} \ea 
    = \ba{cc} {(I_p+{\bf G}^{\mathtt{pt}} {{\bf K}^{\mathtt{md}}_{\bf Q}})^{-1}} & { {(I_p+{\bf G}^{\mathtt{pt}} {{\bf K}^{\mathtt{md}}_{\bf Q}})^{-1}}{{\bf G}^{\mathtt{pt}}}} \\
{{{\bf K}^{\mathtt{md}}_{\bf Q}}{(I_p+{\bf G}^{\mathtt{pt}} {{\bf K}^{\mathtt{md}}_{\bf Q}})^{-1}}} &  {{{\bf K}^{\mathtt{md}}_{\bf Q}}{(I_p+{\bf G}^{\mathtt{pt}} {{\bf K}^{\mathtt{md}}_{\bf Q}})^{-1}}}{{\bf G}^{\mathtt{pt}}} \ea
\end{equation}

\noindent
The LQG cost in \eqref{LQGCostOriginal} is written by using \cite{furieri2020} as below:

\begin{equation*}
\begin{aligned}
\lim\limits_{T \rightarrow \infty} &  \mathbb{E} \bigg[ \dfrac{1}{T} \sum \limits_{t=0}^{T} (y^T_t P_1 y_t + u^T_t P_2 u_t)  \bigg] = \mathbb{E}[y^T_{\infty} P_1 y_{\infty} + u^T_{\infty} P_2 u_{\infty}]\\
 = & \sum \limits_{k=0}^{\infty} \mathbb{E} \bigg[\sigma^2_\nu tr\Big(({{\bf T}^{y \nu}_{\bf Q}})^T P_1 {{\bf T}^{y \nu}_{\bf Q}}\Big) +  \sigma^2_w tr\Big(({{\bf T}^{yw}_{\bf Q}})^2 P_1 {{\bf T}^{yw}_{\bf Q}}  \Big) + \\  & + \sigma^2_\nu tr\Big(({{\bf T}^{u \nu}_{\bf Q}})^2 P_2 {{\bf T}^{u \nu}_{\bf Q}} \Big) + \sigma^2_w tr\Big(({{\bf T}^{u w}_{\bf Q}})^2 P_2 {{\bf T}^{u w}_{\bf Q}} \Big) \bigg]\\
 = &  \sum \limits_{k=0}^{\infty} \Bigg\| \ba{cc} P_1^{1/2} & \\ & P_2^{1/2}  \ea  \ba{cc} {{\bf T}^{y \nu}_{\bf Q}} & {{\bf T}^{yw}_{\bf Q}} \\ {{\bf T}^{u \nu}_{\bf Q}} & {{\bf T}^{u w}_{\bf Q}} \ea \ba{cc} \sigma_\nu I_p & \\ & \sigma_w I_m \ea \Bigg\|^2_F\\
 = & \dfrac{1}{2\pi} \int \limits_{-\pi}^{\pi} \Bigg\| \ba{cc} P_1^{1/2} & \\ & P_2^{1/2}  \ea  \ba{cc} {{\bf T}^{y \nu}_{\bf Q}} & {{\bf T}^{yw}_{\bf Q}} \\ {{\bf T}^{u \nu}_{\bf Q}} & {{\bf T}^{u w}_{\bf Q}} \ea \ba{cc} \sigma_\nu I_p & \\ & \sigma_w I_m \ea \Bigg\|^2_F d\theta\\
 = &  \Bigg\| \ba{cc} P_1^{1/2} & \\ & P_2^{1/2}  \ea  \ba{cc} {{\bf T}^{y \nu}_{\bf Q}} & {{\bf T}^{yw}_{\bf Q}} \\ {{\bf T}^{u \nu}_{\bf Q}} & {{\bf T}^{u w}_{\bf Q}} \ea \ba{cc} \sigma_\nu I_p & \\ & \sigma_w I_m \ea \Bigg\|^2_{\mathcal{H}_2},
\end{aligned}
\end{equation*}
where the second last quality is due to the Parseval's theorem with Frobenius norm and last quality is the definition of the ${\mathcal{H}_2}$ norm. Without loss of generality it is assumed that $P_1 = I_p$, $P_2 = I_m$, $\sigma_\nu = 1$, $\sigma_w = 1$.
Then the square root of the LQG cost can be written as the closed loop responses defined in \eqref{cost5}, \eqref{closedloopresponses2}:
\begin{equation*}
\small
    \mathcal{H}({\bf G}^{\mathtt{pt}}, {{\bf K}^{\mathtt{md}}_{\bf Q}}) =  \Bigg\| \ba{cc} {{\bf T}^{y \nu}_{\bf Q}} & {{\bf T}^{yw}_{\bf Q}} \\ {{\bf T}^{u \nu}_{\bf Q}} & {{\bf T}^{u w}_{\bf Q}} \ea \Bigg\|_{\mathcal{H}_2} = \Bigg\| \ba{cc} {\tilde {\bf Y}}^{\mathtt{md}}_{\bf Q}{\bf \Phi}_{11}^{-1} ({\tilde {\bf M}}^{\mathtt{md}}+\Delta_{\bf \tilde{M}})   &  {\tilde {\bf Y}}^{\mathtt{md}}_{\bf Q}{\bf \Phi}_{11}^{-1} ({\tilde {\bf N}}^{\mathtt{md}}+\Delta_{\bf \tilde{N}}) \\  {\tilde {\bf X}}^{\mathtt{md}}_{\bf Q}{\bf \Phi}_{11}^{-1} ({\tilde {\bf M}}^{\mathtt{md}}+\Delta_{\bf \tilde{M}})  & {\tilde {\bf X}}^{\mathtt{md}}_{\bf Q}{\bf \Phi}_{11}^{-1} ({\tilde {\bf N}}^{\mathtt{md}}+\Delta_{\bf \tilde{N}}) \ea  \Bigg\|_{\mathcal{H}_{2}}.
\end{equation*}

\end{proof}

\begin{proof}{\bf for \propref{LQGcostupperbound2}}: 
Given $\Big\| \ba{cc}  \Delta_{\bf \tilde{M}} &    \Delta_{\bf \tilde{N}}  \ea \Big\|_\infty < \gamma ,  \Bigg\|
\ba{c} {\tilde {\bf Y}}^{\mathtt{md}}_{\bf Q} \\  {\tilde {\bf X}}^{\mathtt{md}}_{\bf Q} \ea \Bigg\|_{\infty} < \dfrac{1}{
\gamma}$ and ${\tilde {\bf M}}^{\mathtt{md}}, {\tilde {\bf N}}^{\mathtt{md}} \in \mathcal{RH}_{\infty}$,
 by using \eqref{triangularinequality}, \eqref{submultiplicativeinfty} and \lemref{submultiplicativeh2} the LQG cost in \eqref{theoremA} is obtained as
\begin{equation*}
\small
    \begin{aligned}
    \mathcal{H}&({\bf G}^{\mathtt{pt}}, {{\bf K}^{\mathtt{md}}_{\bf Q}}) =   \Bigg\|  \ba{c}    {\tilde {\bf Y}}^{\mathtt{md}}_{\bf Q} \\  {\tilde {\bf X}}^{\mathtt{md}}_{\bf Q} \ea {\bf \Phi}_{11}^{-1} \ba{cc} {\tilde {\bf M}^{\mathtt{md}}}+\Delta_{\bf {\tilde{M}}} &  {\tilde {\bf N}^{\mathtt{md}}}+\Delta_{\bf {\tilde{N}}}  \ea \Bigg\|_{\mathcal{H}_{2}}\\ 
     \leq  & \Bigg\| \ba{c}   {\tilde {\bf Y}}^{\mathtt{md}}_{\bf Q} \\  {\tilde {\bf X}}^{\mathtt{md}}_{\bf Q} \ea \Bigg\|_{\mathcal{H}_{2}} \Bigg\|{\bf \Phi}_{11}^{-1} \ba{cc} {\tilde {\bf M}^{\mathtt{md}}}+\Delta_{\bf {\tilde{M}}} &  {\tilde {\bf N}^{\mathtt{md}}}+\Delta_{\bf {\tilde{N}}}  \ea \Bigg\|_{\infty}\\
     \leq &  \Bigg\| \ba{c}    {\tilde {\bf Y}}^{\mathtt{md}}_{\bf Q} \\  {\tilde {\bf X}}^{\mathtt{md}}_{\bf Q} \ea \Bigg\|_{\mathcal{H}_{2}} \Bigg\| {\Bigg({ I_p} - \ba{cc}  -\Delta_{\bf \tilde{M}} &  -\Delta_{\bf \tilde{N}} \ea \ba{c}  {\tilde {\bf Y}}^{\mathtt{md}}_{\bf Q} \\  {\tilde {\bf X}}^{\mathtt{md}}_{\bf Q} \ea \Bigg)}^{-1} \Bigg\|_{\infty} \Bigg\| \bigg( \ba{cc}  {\tilde {\bf M}^{\mathtt{md}}} &  {\tilde {\bf N}^{\mathtt{md}}}  \ea + \ba{cc}  \Delta_{\bf \tilde{M}} &  \Delta_{\bf \tilde{N}} \ea \bigg) \Bigg\|_{\infty}\\
     = &  \Bigg\| \ba{c}    {\tilde {\bf Y}}^{\mathtt{md}}_{\bf Q} \\  {\tilde {\bf X}}^{\mathtt{md}}_{\bf Q} \ea \Bigg\|_{\mathcal{H}_{2}} {\Bigg( \Bigg\| { I_p} + \sum \limits_{j=1}^{\infty} \bigg(\ba{cc}  -\Delta_{\bf \tilde{M}} &  -\Delta_{\bf \tilde{N}} \ea \ba{c}  {\tilde {\bf Y}}^{\mathtt{md}}_{\bf Q} \\  {\tilde {\bf X}}^{\mathtt{md}}_{\bf Q} \ea\bigg)^{j} \Bigg\|_{\infty}\Bigg)} \bigg( \Bigg\| \ba{cc}  {\tilde {\bf M}^{\mathtt{md}}} &  {\tilde {\bf N}^{\mathtt{md}}}  \ea  + \ba{cc}  \Delta_{\bf \tilde{M}} &  \Delta_{\bf \tilde{N}} \ea  \Bigg\|_{\infty}  \bigg)\\
     \leq & \Bigg\| \ba{c}    {\tilde {\bf Y}}^{\mathtt{md}}_{\bf Q} \\  {\tilde {\bf X}}^{\mathtt{md}}_{\bf Q} \ea \Bigg\|_{\mathcal{H}_{2}} {\Bigg( 1 + \Bigg\|\sum \limits_{j=1}^{\infty} \bigg(\ba{cc}  -\Delta_{\bf \tilde{M}} &  -\Delta_{\bf \tilde{N}} \ea  \ba{c}  {\tilde {\bf Y}}^{\mathtt{md}}_{\bf Q} \\  {\tilde {\bf X}}^{\mathtt{md}}_{\bf Q} \ea\bigg)^{j} \Bigg\|_{\infty} \Bigg)} \bigg( \Bigg\| \ba{cc}  {\tilde {\bf M}^{\mathtt{md}}} &  {\tilde {\bf N}^{\mathtt{md}}} \ea \Bigg\|_{\infty} + \gamma \bigg) \\
     = &  \Bigg\| \ba{c}    {\tilde {\bf Y}}^{\mathtt{md}}_{\bf Q} \\  {\tilde {\bf X}}^{\mathtt{md}}_{\bf Q} \ea \Bigg\|_{\mathcal{H}_{2}} {\Bigg[ \Bigg\| \ba{cc}  {\tilde {\bf M}^{\mathtt{md}}} &  {\tilde {\bf N}^{\mathtt{md}}}  \ea \Bigg\|_{\infty} + \gamma} + {\bigg( \sum \limits_{j=1}^{\infty} \gamma^{j} \Bigg\| \ba{c}  {\tilde {\bf Y}}^{\mathtt{md}}_{\bf Q} \\  {\tilde {\bf X}}^{\mathtt{md}}_{\bf Q} \ea \Bigg\|_{\infty}^{j}  \Bigg)} \bigg( \Bigg\| \ba{cc}   {\tilde {\bf M}^{\mathtt{md}}} &  {\tilde {\bf N}^{\mathtt{md}}}  \ea \Bigg\|_{\infty} + \gamma \bigg) \Bigg] \\
     = &  \Bigg\| \ba{c} {\tilde {\bf Y}}^{\mathtt{md}}_{\bf Q} \\  {\tilde {\bf X}}^{\mathtt{md}}_{\bf Q} \ea \Bigg\|_{\mathcal{H}_{2}} {\Bigg[ \Bigg\| \ba{cc}  {\tilde {\bf M}^{\mathtt{md}}} &  {\tilde {\bf N}^{\mathtt{md}}} \ea \Bigg\|_{\infty} + \gamma} + {\dfrac{\gamma \Bigg\| \ba{c}  {\tilde {\bf Y}}^{\mathtt{md}}_{\bf Q} \\  {\tilde {\bf X}}^{\mathtt{md}}_{\bf Q} \ea \Bigg\|_{\infty}}{1 - \gamma \Bigg\| \ba{c}  {\tilde {\bf Y}}^{\mathtt{md}}_{\bf Q} \\  {\tilde {\bf X}}^{\mathtt{md}}_{\bf Q} \ea \Bigg\|_{\infty}}     } \bigg( \Bigg\| \ba{cc}   {\tilde {\bf M}^{\mathtt{md}}} &  {\tilde {\bf N}^{\mathtt{md}}}  \ea \Bigg\|_{\infty} + \gamma \bigg) \Bigg] \\
     \end{aligned}
     \end{equation*}
  \begin{equation*}
\small
    \begin{aligned}
     \leq &  \dfrac{ \Bigg\| \ba{c}  {\tilde {\bf Y}}^{\mathtt{md}}_{\bf Q} \\  {\tilde {\bf X}}^{\mathtt{md}}_{\bf Q} \ea \Bigg\|_{\mathcal{H}_{2}}}{1 - \gamma \Bigg\| \ba{c}  {\tilde {\bf Y}}^{\mathtt{md}}_{\bf Q} \\  {\tilde {\bf X}}^{\mathtt{md}}_{\bf Q} \ea \Bigg\|_{\infty}}   {\Bigg[ \Bigg\| \ba{cc}   {\tilde {\bf M}^{\mathtt{md}}} &  {\tilde {\bf N}^{\mathtt{md}}}   \ea \Bigg\|_{\infty} + \gamma} + {\gamma \Bigg\| \ba{c}   {\tilde {\bf Y}}^{\mathtt{md}}_{\bf Q} \\  {\tilde {\bf X}}^{\mathtt{md}}_{\bf Q} \ea \Bigg\|_{\infty}  } \bigg( \Bigg\| \ba{cc}  {\tilde {\bf M}^{\mathtt{md}}} &  {\tilde {\bf N}^{\mathtt{md}}}  \ea \Bigg\|_{\infty} + \gamma \bigg) \Bigg]\\
     = &  \dfrac{ \Bigg\| \ba{c}  {\tilde {\bf Y}}^{\mathtt{md}}_{\bf Q} \\  {\tilde {\bf X}}^{\mathtt{md}}_{\bf Q} \ea \Bigg\|_{\mathcal{H}_{2}}}{1 - \gamma \Bigg\| \ba{c}  {\tilde {\bf Y}}^{\mathtt{md}}_{\bf Q} \\  {\tilde {\bf X}}^{\mathtt{md}}_{\bf Q} \ea \Bigg\|_{\infty}}   {\Bigg[ \bigg(1 + {\gamma \Bigg\| \ba{c}   {\tilde {\bf Y}}^{\mathtt{md}}_{\bf Q} \\  {\tilde {\bf X}}^{\mathtt{md}}_{\bf Q} \ea \Bigg\|_{\infty}  }\bigg)  \bigg( \Bigg\| \ba{cc}  {\tilde {\bf M}^{\mathtt{md}}} &  {\tilde {\bf N}^{\mathtt{md}}}  \ea \Bigg\|_{\infty} + \gamma \bigg)  \Bigg]}
    \end{aligned}
\end{equation*}

\noindent
Hence the LQG cost in \eqref{theoremA} is upper bounded by:

\begin{align*}
  \mathcal{H}&({\bf G}^{\mathtt{pt}}, {{\bf K}^{\mathtt{md}}_{\bf Q}}) \leq  \dfrac{1}{1 - \gamma \Bigg\| \ba{c}  {\tilde {\bf Y}}^{\mathtt{md}}_{\bf Q} \\  {\tilde {\bf X}}^{\mathtt{md}}_{\bf Q}  \ea \Bigg\|_{\infty}}   {\Bigg[ {h \Big( \gamma, \Bigg\| \ba{c}   {\tilde {\bf Y}}^{\mathtt{md}}_{\bf Q} \\  {\tilde {\bf X}}^{\mathtt{md}}_{\bf Q} \ea \Bigg\|_{\infty} \Big)} \Bigg\| \ba{c}  {\tilde {\bf Y}}^{\mathtt{md}}_{\bf Q} \\  {\tilde {\bf X}}^{\mathtt{md}}_{\bf Q} \ea \Bigg\|_{\mathcal{H}_{2}} \Bigg] }
\end{align*}

\noindent
with ${h \bigg( \gamma, \Bigg\| \ba{c}   {\tilde {\bf Y}}^{\mathtt{md}}_{\bf Q} \\  {\tilde {\bf X}}^{\mathtt{md}}_{\bf Q} \ea \Bigg\|_{\infty} \bigg)}$ = $\bigg(1 + {\gamma \Bigg\| \ba{c}   {\tilde {\bf Y}}^{\mathtt{md}}_{\bf Q} \\  {\tilde {\bf X}}^{\mathtt{md}}_{\bf Q} \ea \Bigg\|_{\infty}  }\bigg)  \bigg( \Bigg\| \ba{cc}  {\tilde {\bf M}^{\mathtt{md}}} &  {\tilde {\bf N}^{\mathtt{md}}}  \ea \Bigg\|_{\infty} + \gamma \bigg)$.




\end{proof}

\section{Quasi Convex Optimization}\label{appendixE}
\noindent
To proof \thmref{theoremB} an important lemma is stated below which is a standard result optimization. 

\begin{lem} \label{sarahopt} (\cite{dean2020})
For functions $f: X \mapsto R$ and $g: X \mapsto R$ and constraint set $C \subseteq X$, consider
\begin{equation}
    \min_{x \in C}  \dfrac{f(x)}{1 - g(x)}
\end{equation}

\noindent
Assuming that $f (x) \geq 0$ and $0 \leq g(x) < 1$, $(\forall) x \in C$, this optimization problem
can be reformulated as an outer single-variable problem and an inner-constrained optimization problem (the objective value of an optimization over the empty set is defined to be infinity):
\begin{equation}
     \min_{x \in C}  \dfrac{f(x)}{1 - g(x)} =  \min_{\delta \in [0,1)} \dfrac{1}{1 - \delta} \min_{x \in C}  \{f(x)| g(x) \leq \delta \}.
\end{equation}
\end{lem}

\begin{proof}\textbf{of \thmref{theoremB}}: 
The LQG cost in \eqref{theoremA} is upper bounded by:

\begin{equation*}
  \mathcal{H}({\bf G}^{\mathtt{pt}}, {{\bf K}^{\mathtt{md}}_{\bf Q}}) \leq  \dfrac{1}{1 - \gamma \Bigg\| \ba{c}  {\tilde {\bf Y}}^{\mathtt{md}}_{\bf Q} \\  {\tilde {\bf X}}^{\mathtt{md}}_{\bf Q}  \ea \Bigg\|_{\infty}}   {\Bigg[ {h \Big( \gamma, \Bigg\| \ba{c}   {\tilde {\bf Y}}^{\mathtt{md}}_{\bf Q} \\  {\tilde {\bf X}}^{\mathtt{md}}_{\bf Q} \ea \Bigg\|_{\infty} \Big)} \Bigg\| \ba{c}  {\tilde {\bf Y}}^{\mathtt{md}}_{\bf Q} \\  {\tilde {\bf X}}^{\mathtt{md}}_{\bf Q} \ea \Bigg\|_{\mathcal{H}_{2}} \Bigg] },
\end{equation*}

\noindent
with ${h \bigg( \gamma, \Bigg\| \ba{c}   {\tilde {\bf Y}}^{\mathtt{md}}_{\bf Q} \\  {\tilde {\bf X}}^{\mathtt{md}}_{\bf Q} \ea \Bigg\|_{\infty} \bigg)}$ = $\bigg(1 + {\gamma \Bigg\| \ba{c}   {\tilde {\bf Y}}^{\mathtt{md}}_{\bf Q} \\  {\tilde {\bf X}}^{\mathtt{md}}_{\bf Q} \ea \Bigg\|_{\infty}  }\bigg)  \bigg( \Bigg\| \ba{cc}  {\tilde {\bf M}^{\mathtt{md}}} &  {\tilde {\bf N}^{\mathtt{md}}}  \ea \Bigg\|_{\infty} + \gamma \bigg)$.



\noindent
Let $f(x) = {h \Big( \gamma, \Bigg\| \ba{c}   {\tilde {\bf Y}}^{\mathtt{md}}_{\bf Q} \\  {\tilde {\bf X}}^{\mathtt{md}}_{\bf Q} \ea \Bigg\|_{\infty} \Big)} \Bigg\| \ba{c}  {\tilde {\bf Y}}^{\mathtt{md}}_{\bf Q} \\  {\tilde {\bf X}}^{\mathtt{md}}_{\bf Q} \ea \Bigg\|_{\mathcal{H}_{2}}$ and $g(x) = \gamma \Bigg\| \ba{c}  {\tilde {\bf Y}}^{\mathtt{md}}_{\bf Q} \\  {\tilde {\bf X}}^{\mathtt{md}}_{\bf Q} \ea \Bigg\|_{\infty}$.
Here for any stable ${\bf Q}$,  $f (x) \geq 0$ and $0 \leq g(x) < 1$. 
Then using \lemref{sarahopt} by re scaling $\delta$  with $\delta/\gamma$ and introducing one additional constraint $\Bigg\| \ba{c} {\tilde {\bf Y}}^{\mathtt{md}}_{\bf Q} \\  {\tilde {\bf X}}^{\mathtt{md}}_{\bf Q}  \ea \Bigg\|_{\infty} \leq \alpha$, 
we get the upper bound of \thmref{theoremB}. 
\end{proof}

\noindent
To check \remref{numerical}, we state the following lemma using \cite{Dumitrescu2007}.

\begin{lem}\label{SDPlemma}
We overload the notation and denote $\ba{c}   {\tilde {\bf Y}}^{\mathtt{md}}_{\bf Q}(z) \\  {\tilde {\bf X}}^{\mathtt{md}}_{\bf Q}(z) \ea$ to be the following truncated system:
\begin{equation}
     \ba{c}   {\tilde {\bf Y}}^{\mathtt{md}}_{\bf Q}(z) \\  {\tilde {\bf X}}^{\mathtt{md}}_{\bf Q}(z) \ea = \sum\limits_{t=0}^N F_t z^{-t}, \quad F_t \in \mathbb{R}^{(p+m)\times p}, (\forall) t.
\end{equation}
Denote $\mathcal{F} = \ba{cccc} F_0 & F_1 & \dots & F_N \ea^T \in \mathbb{R}^{[(p+m)\times N] \times p}$, then the constraint $\Bigg\| \ba{c} {\tilde {\bf Y}}^{\mathtt{md}}_{\bf Q} \\  {\tilde {\bf X}}^{\mathtt{md}}_{\bf Q}  \ea \Bigg\|_{\infty} \leq \gamma$ is equivalent with the existence of a positive semidefinite $S \in \mathcal{S}_+^{(p+m)(N+1)}$, such that
\begin{equation*}
    \ba{cc} S & \mathcal{F}\\ \mathcal{F}^T & \delta I_p \ea \succeq 0, \quad \sum \limits_{i=1}^{N-k} S_{i+k, i} = \delta \xi(k) I_{p+m},\quad k = 0, 1, ..., N.
\end{equation*}
\noindent
where, $ S_{i,k} \in \mathbb{R}^{(p+m)\times (p+m)}$ is the $(i,k)^{\mathrm{th}}$ block of $S$. $\xi(k)$ is the impulse function defined as 
\begin{equation*}
    \xi(k) = \begin{cases} 1, \quad k=0.\\0, \quad otherwise. \end{cases}
\end{equation*}
\end{lem}

\begin{prop}\label{SDPformulation}(SDP formulation for \thmref{theoremB})
We denote the following:

$\ba{c}    {\bf \tilde{Y}}^{\mathtt{md}}_{\bf Q}(z) \\  {\bf \tilde{X}}^{\mathtt{md}}_{\bf Q}(z) \ea = \sum \limits_j F_j z^{-j},  \quad \ba{c}    {\bf \tilde{Y}}^{\mathtt{md}}(z) \\  {\bf \tilde{X}}^{\mathtt{md}}(z) \ea = \sum \limits_j C_j z^{-j},  \quad \ba{c}    {\bf {M}}^{\mathtt{md}}(z) \\  {\bf {N}}^{\mathtt{md}}(z) \ea = \sum \limits_j P_j z^{-j},$\\ with $F_j, C_j, P_j \in \mathbb{R}^{(m+p)\times p}$ and ${\bf Q}(z) = \sum \limits_j Q_j z^{-j},$ with $Q_j \in \mathbb{R}^{p\times p}$.
Let $\mathtt{vec(.)}$ be the column vectorization of a matrix.
Next we define 
$\hat{P}_k = \ba{c:c:c:c} I_p \otimes P_k^T & I_p \otimes P_{k-1}^T & \dots & I_p \otimes P_{k-n}^T \ea$, where $\otimes$ denote the Kronecker product and by convention, $P_{-1} = P_{-2} = \dots = 0$.
$\overline{C}$ and $\overline{q}$ are vectors containing the coefficients of the FIR system $\ba{c}    {\bf \tilde{Y}}^{\mathtt{md}}(z) \\  {\bf \tilde{X}}^{\mathtt{md}}(z) \ea$ and ${\bf Q}(z)$ respectively. Therefore
\begin{equation*}
    \overline{q} = \ba{cccc} \mathtt{vec}(Q_0) & \mathtt{vec}(Q_1) & \dots & \mathtt{vec}(Q_n) \ea^T, \quad \hat{P} = \ba{c} \hat{P}_0 \\ \hdashline  \\ \hat{P}_1 \\ \hdashline  \vdots \\ \hdashline \\\hat{P}_n \ea.
\end{equation*}

\noindent
The inner optimization problem in \thmref{theoremB} after FIR truncation can be expressed as:
\begin{equation} \label{FIRtruncation}
\begin{aligned}
\min_{\epsilon, \overline{q}, S}  \quad & h(\gamma, \alpha) \epsilon \\
\textrm{s.t.} \quad &  S \in \mathcal{S}_+^{(p+m)(n+1)}, \quad \sum \limits_{i=1}^{n-k} S_{i+k, i} &= \min\{\delta, \alpha\} \xi(k) I_{p+m},\quad k = 0:n,\\
& \ba{cc} S & \mathcal{F}\\ \mathcal{F}^T & \min \{\delta, \alpha\} I_p \ea \succeq 0,\\
& \|\overline{C} + \hat{P} \overline{q} \| \leq \epsilon.
\end{aligned}
\end{equation}
\end{prop}


\section{Suboptimality guarantee proof}\label{appendixF}
\begin{proof}\textbf{of \thmref{suboptimalitygurantee}}:\\

\noindent
Let 
${\bf Q_*}$ is the optimal solution for \thmref{theoremB}.
By using \eqref{RealPlantOptimalController1}, it is observed that
\begin{equation*}
\small
    \begin{aligned}
    &\mathcal{H}({\bf G}^{\mathtt{pt}}, {{\bf K}^{\mathtt{md}}_{\bf Q_*}})  =  \Bigg\|   \ba{c}    {\bf \tilde{Y}}^{\mathtt{md}}_{\bf Q_*} \\  {\bf \tilde{X}}^{\mathtt{md}}_{\bf Q_*} \ea {\bf \Phi}_{11}^{-1} {I_p} \ba{cc}  {\tilde {\bf M}^{\mathtt{md}}}+\Delta_{\bf \tilde{M}} &  {\tilde {\bf N}^{\mathtt{md}}}+\Delta_{\bf \tilde{N}}  \ea \Bigg\|_{\mathcal{H}_{2}}\\
    &  = \Bigg\|  \ba{c}    {\bf \tilde{Y}}^{\mathtt{md}}_{\bf Q_*} \\  {\bf \tilde{X}}^{\mathtt{md}}_{\bf Q_*} \ea {\bf \Phi}_{11}^{-1} \ba{cc}  {\tilde {\bf M}^{\mathtt{md}}}+\Delta_{\bf \tilde{M}} &  {\tilde {\bf N}^{\mathtt{md}}}+\Delta_{\bf \tilde{N}}  \ea   \ba{c}  {\tilde {\bf Y}}^{\mathtt{opt}} \\  {\tilde {\bf X}}^{\mathtt{opt}} \ea  \ba{cc}  {\tilde {\bf M}^{\mathtt{md}}}+\Delta_{\bf \tilde{M}} &  {\tilde {\bf N}^{\mathtt{md}}}+\Delta_{\bf \tilde{N}}  \ea \Bigg\|_{\mathcal{H}_{2}}\\
 & \leq \Bigg\|    \ba{c}    {\bf \tilde{Y}}^{\mathtt{md}}_{\bf Q_*} \\  {\bf \tilde{X}}^{\mathtt{md}}_{\bf Q_*} \ea {\bf \Phi}_{11}^{-1}  \ba{cc}  {\tilde {\bf M}^{\mathtt{md}}}+\Delta_{\bf \tilde{M}} &  {\tilde {\bf N}^{\mathtt{md}}}+\Delta_{\bf \tilde{N}}  \ea \Bigg\|_{\infty} \Bigg\|  \ba{c}  {\tilde {\bf Y}}^{\mathtt{opt}} \\  {\tilde {\bf X}}^{\mathtt{opt}} \ea  \ba{cc}  {\tilde {\bf M}^{\mathtt{md}}}+\Delta_{\bf \tilde{M}} &  {\tilde {\bf N}^{\mathtt{md}}}+\Delta_{\bf \tilde{N}}  \ea  \Bigg\|_{\mathcal{H}_{2}}\\
  & \leq \Bigg\| \ba{c}    {\bf \tilde{Y}}^{\mathtt{md}}_{\bf Q_*} \\  {\bf \tilde{X}}^{\mathtt{md}}_{\bf Q_*} \ea {\bf \Phi}_{11}^{-1}  \ba{cc}  {\tilde {\bf M}^{\mathtt{md}}}+\Delta_{\bf \tilde{M}} &  {\tilde {\bf N}^{\mathtt{md}}}+\Delta_{\bf \tilde{N}}  \ea \Bigg\|_{\infty}  \times \mathcal{H}({\bf G}^{\mathtt{pt}}, {\bf K}^{\mathtt{opt}})  \\
\end{aligned}
\end{equation*}

\noindent
Next using same procedure of \propref{LQGcostupperbound2} for $\|.\|_{\infty}$ 
it follows that
\begin{equation*}
    \Bigg\|\ba{c}    {\bf \tilde{Y}}^{\mathtt{md}}_{\bf Q_*} \\  {\bf \tilde{X}}^{\mathtt{md}}_{\bf Q_*} \ea {\bf \Phi}_{11}^{-1}  \ba{cc}  {\tilde {\bf M}^{\mathtt{md}}}+\Delta_{\bf \tilde{M}} &  {\tilde {\bf N}^{\mathtt{md}}}+\Delta_{\bf \tilde{N}}  \ea \Bigg\|_{\infty} \leq \dfrac{1}{1 - \gamma \bigg\| \ba{c} {\bf \tilde{Y}}^{\mathtt{md}}_{\bf Q_*} \\  {\bf \tilde{X}}^{\mathtt{md}}_{\bf Q_*} \ea \Bigg\|_{\infty}} \times {g \bigg( \gamma, \Bigg\| \ba{c}   {\tilde {\bf Y}}^{\mathtt{md}}_{\bf Q_*} \\  {\tilde {\bf X}}^{\mathtt{md}}_{\bf Q_*} \ea \Bigg\|_{\infty} \bigg)}
\end{equation*}

\noindent
with $\small {g \bigg( \gamma, \Bigg\| \ba{c}   {\tilde {\bf Y}}^{\mathtt{md}}_{\bf Q_*} \\  {\tilde {\bf X}}^{\mathtt{md}}_{\bf Q_*} \ea \Bigg\|_{\infty} \bigg)} = \Bigg\| \ba{c}  {\tilde {\bf Y}}^{\mathtt{md}}_{\bf Q_*} \\  {\tilde {\bf X}}^{\mathtt{md}}_{\bf Q_*} \ea \Bigg\|_{\infty} \bigg(1 + {\gamma \Bigg\| \ba{c}   {\tilde {\bf Y}}^{\mathtt{md}}_{\bf Q_*} \\  {\tilde {\bf X}}^{\mathtt{md}}_{\bf Q_*} \ea \Bigg\|_{\infty}  }\bigg)  \bigg( \Bigg\| \ba{cc}  {\tilde {\bf M}^{\mathtt{md}}} &  {\tilde {\bf N}^{\mathtt{md}}}  \ea \Bigg\|_{\infty} + \gamma \bigg)$.\\

\noindent
\begin{equation*}
\scriptsize
    \mathcal{H}({\bf G}^{\mathtt{pt}}, {{\bf K}^{\mathtt{md}}_{\bf Q_*}})^2 -  \mathcal{H}({\bf G}^{\mathtt{pt}}, {\bf K}^{\mathtt{opt}})^2  \leq \Bigg[ \Bigg( \dfrac{1}{1 - \gamma \Bigg\| \ba{c} {\bf \tilde{Y}}^{\mathtt{md}}_{\bf Q_*} \\  {\bf \tilde{X}}^{\mathtt{md}}_{\bf Q_*} \ea \Bigg\|_{\infty}} \times {g \bigg( \gamma, \Bigg\| \ba{c}   {\tilde {\bf Y}}^{\mathtt{md}}_{\bf Q_*} \\  {\tilde {\bf X}}^{\mathtt{md}}_{\bf Q_*} \ea \Bigg\|_{\infty} \bigg)} \Bigg)^2 -1 \Bigg]  \mathcal{H}({\bf G}^{\mathtt{pt}}, {\bf K}^{\mathtt{opt}})^2.
\end{equation*}
\vspace{5pt}
\noindent
Therefore the relative error on the LQG cost becomes
\useshortskip
\begin{equation*}\label{relativeerror3}
\begin{aligned}
\dfrac{ \mathcal{H}({\bf G}^{\mathtt{pt}}, {{\bf K}^{\mathtt{md}}_{\bf Q_*}})^2 -  \mathcal{H}({\bf G}^{\mathtt{pt}}, {\bf K}^{\mathtt{opt}})^2 }{ \mathcal{H}({\bf G}^{\mathtt{pt}}, {\bf K}^{\mathtt{opt}})^2 } & \leq \Bigg[ \dfrac{1}{1 - \gamma \Bigg\| \ba{c} {\bf \tilde{Y}}^{\mathtt{md}}_{\bf Q_*} \\  {\bf \tilde{X}}^{\mathtt{md}}_{\bf Q_*} \ea \Bigg\|_{\infty}} \times {g \bigg( \gamma, \Bigg\| \ba{c}   {\tilde {\bf Y}}^{\mathtt{md}}_{\bf Q_*} \\  {\tilde {\bf X}}^{\mathtt{md}}_{\bf Q_*} \ea \Bigg\|_{\infty} \bigg)} \Bigg]^2 -1,\\
& = \dfrac{1}{\Big(1-\gamma \Bigg\| \ba{c} {\bf \tilde{Y}}^{\mathtt{md}}_{\bf Q_*} \\  {\bf \tilde{X}}^{\mathtt{md}}_{\bf Q_*} \ea \Bigg\|_{\infty} \Big)^2} \times {g \bigg( \gamma, \Bigg\| \ba{c}   {\tilde {\bf Y}}^{\mathtt{md}}_{\bf Q_*} \\  {\tilde {\bf X}}^{\mathtt{md}}_{\bf Q_*} \ea \Bigg\|_{\infty} \bigg)}^2 -1.\\
\end{aligned}
\end{equation*}


\begin{equation*}
    \textrm{with} \hspace{2pt}
    \small {g \bigg( \gamma, \Bigg\| \ba{c}   {\tilde {\bf Y}}^{\mathtt{md}}_{\bf Q_*} \\  {\tilde {\bf X}}^{\mathtt{md}}_{\bf Q_*} \ea \Bigg\|_{\infty} \bigg)} = \Bigg\| \ba{c}  {\tilde {\bf Y}}^{\mathtt{md}}_{\bf Q_*} \\  {\tilde {\bf X}}^{\mathtt{md}}_{\bf Q_*} \ea \Bigg\|_{\infty} \bigg(1 + {\gamma \Bigg\| \ba{c}   {\tilde {\bf Y}}^{\mathtt{md}}_{\bf Q_*} \\  {\tilde {\bf X}}^{\mathtt{md}}_{\bf Q_*} \ea \Bigg\|_{\infty}  }\bigg)  \bigg( \bigg\| \ba{cc}  {\tilde {\bf M}^{\mathtt{md}}} &  {\tilde {\bf N}^{\mathtt{md}}}  \ea \bigg\|_{\infty} + \gamma \bigg) .
\end{equation*}
\end{proof}

\section{Non-asymptotic System Identification }\label{appendixG}

\subsection{Noise Contaminated Plant}\label{appendixG1}

\begin{proof}\textbf{for \lemref{noisecontaminatedplant}}:
The proof is started by writing the Bezout identity for Dual Youla Parametrization using \thmref{DualYoulaaa},
\begin{equation}\label{dualyoulabezoutidentity}
\begin{split}
\ba{cc}  \tilde{\bf M}^{\mathtt{md}}-{\bf R}^{\mathtt{md}}{\bf X}^{\mathtt{md}} & \tilde{\bf N}^{\mathtt{md}}+{\bf R}^{\mathtt{md}}{\bf Y}^{\mathtt{md}}\\ - {\bf X}^{\mathtt{md}} &  {\bf Y}^{\mathtt{md}} \ea
\ba{cc}  \tilde {\bf Y}^{\mathtt{md}} & -({\bf N}^{\mathtt{md}} +  \tilde{\bf Y}^{\mathtt{md}}{\bf R}^{\mathtt{md}}) \\   \tilde {\bf X}^{\mathtt{md}} &  {\bf M}^{\mathtt{md}} - \tilde{\bf X}^{\mathtt{md}}{\bf R}^{\mathtt{md}} \ea = I_{p+m},\\
\ba{cc}   \tilde {\bf Y}^{\mathtt{md}} & -({\bf N}^{\mathtt{md}} +  \tilde{\bf Y}^{\mathtt{md}}{\bf R}^{\mathtt{md}}) \\   \tilde {\bf X}^{\mathtt{md}} &  {\bf M}^{\mathtt{md}} - \tilde{\bf X}^{\mathtt{md}}{\bf R}^{\mathtt{md}} \ea 
\ba{cc}   \tilde{\bf M}^{\mathtt{md}}-{\bf R}^{\mathtt{md}}{\bf X}^{\mathtt{md}} & \tilde{\bf N}^{\mathtt{md}}+{\bf R}^{\mathtt{md}}{\bf Y}^{\mathtt{md}}\\ - {\bf X}^{\mathtt{md}} &  {\bf Y}^{\mathtt{md}} \ea
 =  I_{p+m}.
\end{split}
\end{equation}

\noindent
Next the closed loop maps with $e_1$ is checked in \figref{clsysid2}.
First, for the left part  $e_1 = ({\bf M}^{\mathtt{md}})^{-1} \Big[u + {\bf \tilde{X}}^{\mathtt{md}} \Big( {\bf R}^{\mathtt{md}}e_1 + ({\bf \tilde{M}}^{\mathtt{md}} -  {\bf R}^{\mathtt{md}}{\bf X}^{\mathtt{md}})\nu \Big) \Big]$,
such that $ ({\bf M}^{\mathtt{md}} - {\bf \tilde{X}}^{\mathtt{md}} {\bf R}^{\mathtt{md}})e_1 = u + {\bf \tilde{X}}^{\mathtt{md}} ({\bf \tilde{M}}^{\mathtt{md}} -  {\bf R}^{\mathtt{md}}{\bf X}^{\mathtt{md}})\nu$,
therefore we get 
\useshortskip
\begin{equation}\label{sysidxxx}
    e_1 = ({\bf M}^{\mathtt{md}} - {\bf \tilde{X}}^{\mathtt{md}} {\bf R}^{\mathtt{md}})^{-1} \Big[ u + {\bf \tilde{X}}^{\mathtt{md}} ({\bf \tilde{M}}^{\mathtt{md}} -  {\bf R}^{\mathtt{md}}{\bf X}^{\mathtt{md}})\nu \Big].
\end{equation}

\noindent
Then for the right part $ y = {\bf N}^{\mathtt{md}}e_1 + {\bf \tilde{Y}}^{\mathtt{md}} \Big( {\bf R}^{\mathtt{md}}e_1 + ({\bf \tilde{M}}^{\mathtt{md}} -  {\bf R}^{\mathtt{md}}{\bf X}^{\mathtt{md}})\nu \big)$, 
such that $y = ({\bf N}^{\mathtt{md}} + {\bf \tilde{Y}}^{\mathtt{md}} {\bf R}^{\mathtt{md}})e_1 + {\bf \tilde{Y}}^{\mathtt{md}} ({\bf \tilde{M}}^{\mathtt{md}} -  {\bf R}^{\mathtt{md}}{\bf X}^{\mathtt{md}})\nu$.\\
\noindent
Next using \eqref{dualyoulabezoutidentity} and \eqref{sysidxxx} it follows that
\useshortskip
\begin{equation*}
\small
    \begin{aligned}
    y &=  ({\bf N}^{\mathtt{md}} + {\bf \tilde{Y}}^{\mathtt{md}} {\bf R}^{\mathtt{md}}) ({\bf M}^{\mathtt{md}} - {\bf \tilde{X}}^{\mathtt{md}} {\bf R}^{\mathtt{md}})^{-1} \Big[ u + {\bf \tilde{X}}^{\mathtt{md}} ({\bf \tilde{M}}^{\mathtt{md}} -  {\bf R}^{\mathtt{md}}{\bf X}^{\mathtt{md}})\nu \Big] + {\bf \tilde{Y}}^{\mathtt{md}} ({\bf \tilde{M}}^{\mathtt{md}} -  {\bf R}^{\mathtt{md}}{\bf X}^{\mathtt{md}})\nu\\
     &= ({\bf N}^{\mathtt{md}} + {\bf \tilde{Y}}^{\mathtt{md}} {\bf R}^{\mathtt{md}}) ({\bf M}^{\mathtt{md}} - {\bf \tilde{X}}^{\mathtt{md}} {\bf R}^{\mathtt{md}})^{-1}  u +\\ & + ({\bf N}^{\mathtt{md}} + {\bf \tilde{Y}}^{\mathtt{md}} {\bf R}^{\mathtt{md}}) ({\bf M}^{\mathtt{md}} - {\bf \tilde{X}}^{\mathtt{md}} {\bf R}^{\mathtt{md}})^{-1} {\bf \tilde{X}}^{\mathtt{md}} ({\bf \tilde{M}}^{\mathtt{md}} -  {\bf R}^{\mathtt{md}}{\bf X}^{\mathtt{md}})\nu   + {\bf \tilde{Y}}^{\mathtt{md}} ({\bf \tilde{M}}^{\mathtt{md}} -  {\bf R}^{\mathtt{md}}{\bf X}^{\mathtt{md}})\nu\\
     &= {\bf G}_{\bf R}^{\mathtt{md}} u + ({\bf N}^{\mathtt{md}} + {\bf \tilde{Y}}^{\mathtt{md}} {\bf R}^{\mathtt{md}}) ({\bf M}^{\mathtt{md}} - {\bf \tilde{X}}^{\mathtt{md}} {\bf R}^{\mathtt{md}})^{-1} (\tilde{\bf M}^{\mathtt{md}}-{\bf R}^{\mathtt{md}}{\bf X}^{\mathtt{md}}) {\bf {X}}^{\mathtt{md}} \nu   + {\bf \tilde{Y}}^{\mathtt{md}} ({\bf \tilde{M}}^{\mathtt{md}} -  {\bf R}^{\mathtt{md}}{\bf X}^{\mathtt{md}})\nu\\
     &= {\bf G}_{\bf R}^{\mathtt{md}} u +  ({\bf N}^{\mathtt{md}}+{\bf \tilde{Y}}^{\mathtt{md}} {\bf R}^{\mathtt{md}})  {\bf {X}}^{\mathtt{md}} \nu  + {\bf \tilde{Y}}^{\mathtt{md}} ({\bf \tilde{M}}^{\mathtt{md}} -  {\bf R}^{\mathtt{md}}{\bf X}^{\mathtt{md}})\nu\\
    &= {\bf G}_{\bf R}^{\mathtt{md}} u + \Big(  {\bf \tilde{Y}}^{\mathtt{md}} ({\bf \tilde{M}}^{\mathtt{md}} -  {\bf R}^{\mathtt{md}}{\bf X}^{\mathtt{md}}) + ({\bf N}^{\mathtt{md}}+{\bf \tilde{Y}}^{\mathtt{md}} {\bf R}^{\mathtt{md}})  {\bf {X}}^{\mathtt{md}}  \Big)\nu =  {\bf G}_{\bf R}^{\mathtt{md}} u + \nu.
    \end{aligned}
\end{equation*}
\end{proof}

\noindent
It can observed from \figref{clsysid2} and from proof of \lemref{noisecontaminatedplant} that
\begin{equation} \label{xysignals}
\begin{split}
    ({\bf M}^{\mathtt{md}} - {\bf \tilde{X}}^{\mathtt{md}} {\bf R}^{\mathtt{md}})e_1 & = u + {\bf \tilde{X}}^{\mathtt{md}} ({\bf \tilde{M}}^{\mathtt{md}} -  {\bf R}^{\mathtt{md}}{\bf X}^{\mathtt{md}})\nu \\
   & = w + ({\bf Y}^{\mathtt{md}})^{-1} {\bf X}^{\mathtt{md}} r  - ({\bf Y}^{\mathtt{md}})^{-1} {\bf X}^{\mathtt{md}} y  + {\bf \tilde{X}}^{\mathtt{md}}({\bf \tilde{M}}^{\mathtt{md}} - \mathbf{R}{\bf \tilde{X}}^{\mathtt{md}})\nu; \\\vspace{4pt}
      ({\bf N}^{\mathtt{md}} + {\bf \tilde{Y}}^{\mathtt{md}} {\bf R}^{\mathtt{md}})e_1 & = y - {\bf \tilde{Y}}^{\mathtt{md}} ({\bf \tilde{M}}^{\mathtt{md}} -  {\bf R}^{\mathtt{md}}{\bf X}^{\mathtt{md}})\nu.
\end{split}
\end{equation}

\noindent
Multiplying first part of \eqref{xysignals} by ${\bf Y}^{\mathtt{md}}$ and the second part of \eqref{xysignals} by ${\bf X}^{\mathtt{md}} $ it follows that
\begin{equation}\label{sysidxx}
e_1 = {\bf X}^{\mathtt{md}} r + {\bf Y}^{\mathtt{md}} w
\end{equation}
\noindent
Furthermore
\begin{equation}\label{sysidmx}
{\bf M}^{\mathtt{md}} e_1 = u +{\bf \tilde{X}}^{\mathtt{md}} e_2;  \hspace{5pt}
{\bf N}^{\mathtt{md}} e_1 = y - {\bf \tilde{Y}}^{\mathtt{md}} e_2
\end{equation}
\noindent
Multiplying first part of \eqref{sysidmx} by ${\bf \tilde{N}}^{\mathtt{md}}$ and second part of \eqref{sysidmx} by ${\bf \tilde{M}}^{\mathtt{md}}$ and subtracting it is obtained that
\begin{equation}\label{z}
e_2 = {\bf \tilde{M}}^{\mathtt{md}} y - {\bf \tilde{N}}^{\mathtt{md}} u
\end{equation}
Hence $e_1$ and $e_2$ can be measured directly from \eqref{sysidxx} and \eqref{z} respectively.

\subsection{Choice of Noise}

\begin{lem}
Let ${\bf H}$ be a discrete time dynamic system with input ${\bf X}$ and output ${\bf Y}$ s.t.
${\bf Y}(z) = {\bf H}(z) {\bf X}(z)$.
Let ${x}[n]$ be a stationary stochastic process with mean function $m_{x}[n]$ and spectral density $\phi_{\bf X}(z)$. Similarly,  ${y}[n]$ be a stationary stochastic process with mean function $m_{y}[n]$ and spectral density $\phi_{\bf Y}(z)$ Then $m_{y}[n] = {\bf H}(z)\Big|_{z=1} m_{x}[n]$ and
$\phi_{\bf Y}(z) = {\bf H}(z) \phi_{\bf X}(z) {\bf H}^T(z^{-1})$.
\end{lem}

\begin{lem}[Spectral Factorization]\cite{robustcontrol}
Given any arbitrary rational matrix ${\bf G}(z)$, one can always always find a rational matrix ${\bf V}(z)$, such that\\
\begin{equation}
    {\bf G}(z) {\bf G}^T(z^{-1}) = {\bf V}(z) {\bf V}^T(z^{-1}),
\end{equation}
where ${\bf V}$ is analytic outside the closed unit circle including at infinity.
\end{lem}

For the identification algorithm, it is needed that $\phi_{{\bf \tilde{M}}^{\mathtt{md}} \nu} (z) = \dfrac{1}{2\pi}I$. Therefore 
$
{\bf \tilde{M}}^{\mathtt{md}}(z) \phi_\nu (z) \times $
$({{\bf \tilde{M}}^{\mathtt{md}}})^T(z^{-1}) = \dfrac{1}{2\pi}I
$
such that $\phi_\nu (z)  = \dfrac{1}{2\pi}  ({{\bf \tilde{M}}^{\mathtt{md}}})^{-1}(z) ({{\bf \tilde{M}}^{\mathtt{md}}})^{-T}(z^{-1})$.
However, $\phi_\nu (z)$ may not necessarily be a density function for a stochastic process since the stability of $({\bf \tilde{M}}^{\mathtt{md}})^{-1}$ cannot be guaranteed. Therefore  we consider the spectral factorization of $\phi_{\nu}(z)$ as:
\begin{equation}
     ({\bf \tilde{M}}^{\mathtt{md}})^{-1}(z) ({\bf \tilde{M}}^{\mathtt{md}})^{-T}(z^{-1}) = I_\nu (z) I_\nu^T (z^{-1}),
\end{equation}
where $I_\nu$ has no poles outside the closed unit circle including at infinity. Next we choose $\{ \nu[n] \}_{n=1}^T$ to be a Gaussian process with zero mean and spectral density $\dfrac{1}{2\pi} I_\nu (z) I_\nu^T (z^{-1})$.
Furthermore,  as the input $u$ = ${\bf X}^{\mathtt{md}} r + {\bf Y}^{\mathtt{md}} w - {\bf X}^{\mathtt{md}} \nu$, and we want to make $\{ u[n] \}$ isotropic, consider the following:
 generate $\{ r[n] \}_{n=1}^T = \{ \nu[n] \}_{n=1}^T$, then $u = {\bf Y}^{\mathtt{md}} w$ such that  ${\bf Y}^{\mathtt{md}}(z) \phi_w (z) ({\bf Y}^{\mathtt{md}})^T(z^{-1}) = \dfrac{1}{2\pi}I$ with  $ \phi_w (z)  = \dfrac{1}{2\pi}  ({\bf Y}^{\mathtt{md}})^{-1}(z) ({\bf Y}^{\mathtt{md}})^{-T}(z^{-1})$. Next compute  $I_w(z)$ by spectral factorization such that
\begin{equation*}
     ({\bf Y}^{\mathtt{md}})^{-1}(z) ({\bf Y}^{\mathtt{md}})^{-T}(z^{-1}) = I_w (z) I_w^T (z^{-1}),
\end{equation*}
and $I_w(z)$ has no poles outside the closed unit circle including at infinity. Finally, we choose $\{ r[n] \}_{n=1}^T = \{ \nu[n] \}_{n=1}^T$ and $\{ w[n] \}_{n=1}^T$ as Gaussian processes with  zero mean  and spectral density $\dfrac{1}{2\pi}  I_w (z) I_w^T (z^{-1})$.

\subsection{Identification Preliminaries}
\noindent
Let's define an LTI system with dynamics as \eqref{stateeq} by its Markov parameters $\mathcal{M} = (C, A, B)$. 
For any matrix $A$, $\sigma_i(A)$ is defined as the $i^{th}$ singular value of $A$ with $\sigma_i(A) \geq \sigma_{i+1}(A)$. 
Similarly let $\rho_i(A)$ = $|\lambda_i(A)|$, where $\lambda_i(A)$ is the eigenvalue of $A$ with $\rho_i(A) \geq \rho_{i+1}(A)$.

\noindent
\textbf{Definition}: A matrix $A \in \mathbb{R}^{n\times n}$ is stable if $\rho(A) < 1$.

\noindent
\textbf{Definition}: The transfer function of $\mathcal{M} = (C, A, B)$ is given by $G(z) = C(zI - A)^{-1}B$, where $z \in \mathbb{C}$.

\noindent
\textbf{Definition}: The $(k, p, q)$ - dimensional Hankel matrix for $\mathcal{M} = (C, A, B)$ is defined as,
\begin{equation}
\small
    \mathcal{H}_{k,p,q}(R) = \begin{bmatrix} C_RA_R^kB_R & C_RA_R^{k+1}B_R & \dots &C_RA_R^{q+k-1}B_R \\C_RA_R^{k+1}B_R & C_RA_R^{k+2}B_R & \dots & C_RA_R^{q+k}B_R\\ \vdots & \vdots & \ddots & \vdots \\  C_RA_R^{p+k-1}B_R & \dots & \dots & C_RA_R^{p+q+k-2}B_R\\
\end{bmatrix} = \mathcal{H}_{k,p,q} \in \mathbb{R}^{p\times q}
\end{equation}

\noindent
and its associated Toeplitz matrix is defined as
$\mathcal{M}$ = $(C, A, B)$ is defined as,
\begin{equation}
\small
\mathcal{T}_{k,d}(R) = \begin{bmatrix} 0 & 0 & \dots & 0 & 0  \\ C_RA_R^{k}B_R & 0 & \dots & 0 & 0\\ \vdots & \vdots & \ddots  & \vdots & \vdots \\ C_RA_R^{d+k-3}B_R & \dots & C_RA_R^kB_R  & 0 & 0\\ C_RA_R^{d+k-2}B_R & C_RA_R^{d+k-3}B_R & \dots & C_RA_R^kB_R & 0\\ 
\end{bmatrix} = \mathcal{T}_{k,d}\in \mathbb{R}^{d\times d}
\end{equation}
\noindent
Also we define the Toeplitz matrix corresponding to the process noise as,
\begin{equation}
\small
    \mathcal{TO}_{k,d}(R) = \begin{bmatrix} 0 & 0 & \dots & 0 & 0  \\ C_RA_R^{k} & 0 & \dots & 0 & 0\\ \vdots & \vdots & \ddots  & \vdots & \vdots \\ C_RA_R^{d+k-3} & \dots &  C_RA_R^k  & 0 & 0\\ C_RA_R^{d+k-2} & C_RA_R^{d+k-3} & \dots & C_RA_R^k & 0\\ \end{bmatrix} = \mathcal{TO}_{k,d} \in \mathbb{R}^{p\times q}
\end{equation}


\noindent
Let's also define, $\Bar{z}_{l,d}^{+}$ = $ \ba{cccc} z_{l}\\ z_{l+1}\\ \vdots \\ z_{l+d-1}  \ea$  and $\Bar{z}_{l,d}^{-}$ = $\ba{cccc} z_l\\ z_{l-1}\\ \vdots \\ z_{l-d+1}  \ea$ \\





\noindent
\subsection{Identification Algorithms }

\noindent
The algorithms for System Identification using the result from \cite{Dahleh2020} are stated below: 

\begin{algorithm}[H]\label{identificationalgorithm}
\SetKwInOut{Input}{Input}
\Input{$T = $ Time horizon for learning, $d =$ Hankel matrix size, $m =$ Input dimension, $p =$ Output dimension} 
\SetAlgoLined
\KwData{signal $\{r_j\}_{j=1}^{2T}, \{w_j\}_{j=1}^{2T}$, signal $\{y_j\}_{j=1}^{2T}, \{u_j\}_{j=1}^{2T}$ }
\KwResult{Estimated Hankel Matrix for ${\bf R}^{\mathtt{md}}$, $\mathcal{\hat{H}}_{0,d,d}$}
\BlankLine
 Generate $2T$ noises, $\{r_t\}_{t=1}^{2T}, \{w_t\}_{t=1}^{2T}$, $\{\nu_t\}_{t=1}^{2T}$ \;
  Select input, $\{u_t\}_{t=1}^{2T} =  X\{r_t\}_{t=1}^{2T} + Y\{w_t\}_{t=1}^{2T} - X\{\nu_t\}_{t=1}^{2T} \in \mathbb{R}^{m \times 2T}$\;
  Collect $2T$ input-output pairs, $\{u_t, z_t\}_{t=1}^{2T}$\;
  Do $\mathcal{\hat{H}}_{0,d,d} = \arg \min \limits_{\mathcal{H}} \sum \limits_{l=0}^{T-1} \Big\|\bar{z}_{l+d+1,d}^{+} - \mathcal{H}\bar{u}_{l+d,d}^{-} \Big\|_2^2 $\;
  Return $\mathcal{\hat{H}}_{0,d,d} $.
\caption{Derive Hankel Matrix for ${\bf R}^{\mathtt{md}}$}
\end{algorithm}
\vspace{5pt}

\noindent
By running Algorithm-\ref{identificationalgorithm} for $T$ times the desired $\mathcal{\hat{H}}_{0,d,d}$ is obtained.
It can be shown that, 
\begin{equation*}
    \mathcal{\hat{H}}_{0,d,d} = \bigg( \sum\limits_{l=0}^{T-1} \bar{z}_{l+d+1,d}^{+} (\bar{U}_{l+d,d}^{-})^{T} \bigg) \bigg( \sum\limits_{l=0}^{T-1} \bar{U}_{l+d,d}^{-} (\bar{U}_{l+d,d}^{-})^{T}  \bigg)
\end{equation*}

\noindent
Parameters needed before Algorithm-\ref{identificationalgorithm2}:
$m$: Input Dimension, $p$: Output Dimension, $\delta $:  Error probability,
$\beta $: Upper bound on $\|\mathcal{R}\|_{\infty}$, $c, C$: Known absolute constants, $\mathcal{R}$: upper bound on  $\dfrac{\|\mathcal{TO}_{k,d}\|_2}{\|\mathcal{T}_{0,\infty}\|_2}$.\\

\begin{algorithm}[H]\label{identificationalgorithm2}
\SetKwInOut{Output}{Output}
\Output{$\hat{d}, \mathcal{\hat{H}}_{0,d,d}$}
\SetAlgoLined
\BlankLine
  Define $\mathcal{D}(T) = \{ d \in \mathbb{N} | d \leq \dfrac{T}{cm^2 log^3(Tm/\delta)} \} $, $\alpha(h) = \sqrt{h}.\big(\sqrt{\dfrac{m+hp+log(T/\delta)}{T}}\big)$\;
  Define $d_0(T, \delta) = \inf \limits_{l \in \mathbb{N}} \{l \big| \Big\| \mathcal{\hat{H}}_{0,l,l} - \mathcal{\hat{H}}_{0,h,h}\Big\|_2 \leq 16\beta R(\alpha (h) + 2 \alpha (l) ), \forall h \in D(T), h \geq l \}$\;
  Choose $\hat{d} = \max \{d_0(T, \delta), \left \lceil(log(T/\delta)) \right \rceil \} $ \;
  Return $\hat{d}$, $\mathcal{\hat{H}}_{0,\hat{d},\hat{d}}$.
\caption{Choice of $d$}
\end{algorithm}

\vspace{10pt}

\noindent
{\bf Remarks}: $d_*(T,\delta), T_*(\delta)$ is defined in Subsection-\ref{probabilisticgurantee}. $d_*(T,\delta)$ depends on $ \Big\| \mathcal{\hat{H}}_{0,\hat{d},\hat{d}} - \mathcal{\hat{H}}_{0,\infty,\infty}\Big\|_2 $, which is unknown. 
\noindent
In such case, $\hat{d}$ can be an approximation for $d_*(T,\delta)$.
Algorithm-\ref{identificationalgorithm2} does not require any prior knowledge on $d_*(T,\delta)$, the unknown parameter only appear in $T_*(\delta)$ in order to give a theoretical guarantee.
Finally Algorithm-\ref{identificationalgorithm3} is described at the top of next page to retrieve the state space realization parameters of ${\bf R}(z)$ , which is a balanced truncation on $\mathcal{\hat{H}}_{0,\hat{d},\hat{d}}$ with $\hat{d}$ obtained from Algorithm-\ref{identificationalgorithm2}.

\vspace{10pt}



\begin{algorithm}[H]\label{identificationalgorithm3}
\SetKwInOut{Input}{Input}

\Input{$T = $ Time horizon for learning, $\hat{d} =$ Hankel size dimension, $m =$ Input dimension, $p =$ Output dimension} 
\SetKwInOut{Output}{Output}
\Output{System Parameters: $(\hat{C}_{\hat{d}}, \hat{A}_{\hat{d}}, \hat{B}_{\hat{d}})$}
\SetAlgoLined
\BlankLine
  $\mathcal{H} = \mathcal{H}_{0, \hat{d}, \hat{d}}$\;
  Pad $\mathcal{H}$ with zeros to ake of dimension $4p \hat{d} \times 4 m \hat{d} $\;
  $U, \Sigma, V \xleftarrow{}$ SVD of $\mathcal{H}$\;
  $U_{\hat{d}}$, $V_{\hat{d}} \xleftarrow{}$ top $\hat{d}$ singular vectors\;
  $\hat{C}_{\hat{d}} \xleftarrow{}$ first $p$ rows of $U_{\hat{d}} \Sigma_{\hat{d}}^{1/2}$\;
  $\hat{B}_{\hat{d}} \xleftarrow{}$ first $m$ columns of $\Sigma_{\hat{d}}^{1/2} V_{\hat{d}}^T$\;
  $Z_0 = [ U_{\hat{d}} \Sigma_{\hat{d}}^{1/2} ]_{1:4p\hat{d}-p} $, $Z_1 = [\Sigma_{\hat{d}}^{1/2} V_{\hat{d}}^T]_{p+1} $\;
  $\hat{C}_{\hat{d}} \xleftarrow{} (Z_0^T Z_0)^{-1} Z_0^T Z_1$\;
  Return $(\hat{C}_{\hat{d}}, \hat{A}_{\hat{d}}, \hat{B}_{\hat{d}})$.
\caption{Retrieving Matrices $A$, $B$, $C$ of ${\bf R}^{\mathtt{md}}$}
\end{algorithm}

\subsection{End to End Complexity Proof}\label{appendixG5}

\begin{proof}\textbf{of \lemref{HRbound}}:
    Let ${{ G}^{\mathtt{md}}_i}$ represents the $i^{th}$ Markov parameter of ${\bf R}^{\mathtt{md}}$, i.e.   ${{ G}^{\mathtt{md}}_i}= {C}^{\mathtt{md}}({A}^{\mathtt{md}})^i {B}^{\mathtt{md}}$ \\ and ${{\bf G}^{\mathtt{pt}}_i}$ represents the $i^{th}$ Markov parameter of ${\bf R}^{\mathtt{pt}}$, i.e.   ${{ G}^{\mathtt{pt}}_i} = {C}^{\mathtt{pt}}({A}^{\mathtt{pt}})^i{B}^{\mathtt{pt}}$.
Then ${\bf R}^{\mathtt{md}} - {\bf R}^{\mathtt{pt}}$ = $\sum \limits_{i=0}^{\infty} ({{G}^{\mathtt{md}}_i} - {{G}^{\mathtt{pt}}_i})z^{-i}$.
For a fixed $i$, $\Big\| ({{G}^{\mathtt{md}}_i} - {{G}^{\mathtt{pt}}_i})z^{-i} \Big\|_{\infty}^2$ = $\max\limits_{\sum \limits_{k=0}^{\infty} u_k^T u_k = 1} \sum \limits_{k=0}^{\infty} y_k^T y_k$ = $\max\limits_{\sum \limits_{k=0}^{\infty} u_k^T u_k = 1} \sum \limits_{k=0}^{\infty} \Big\| ({{G}^{\mathtt{md}}_i} - {{G}^{\mathtt{pt}}_i})u_{k-i} \Big\|_{2}^2$, 
such that $\Big\| ({{G}^{\mathtt{md}}_i} - {{G}^{\mathtt{pt}}_i})z^{-i} \Big\|_{\infty}^2 \leq \|({{G}^{\mathtt{md}}_i} - {{G}^{\mathtt{pt}}_i})\|_2^2 \max\limits_{\sum \limits_{k=0}^{\infty} u_k^T u_k = 1} \sum \limits_{k=0}^{\infty} \Big\| u_{k-i} \Big\|_{2}^2$ = $ \|({{G}^{\mathtt{md}}_i} - {{G}^{\mathtt{pt}}_i})\|_2^2 $,
by defining  $u_{-1} = u_{-2} = u_{-3} = .... = 0$ for convenience. \\ Furthermore
$\Big\| {\bf R}^{\mathtt{md}}  - {\bf R}^{\mathtt{pt}} \Big\|_{\infty}^2 = \Big\|  \sum \limits_{i=0}^{\infty} ({{G}^{\mathtt{md}}_i} - {{G}^{\mathtt{pt}}_i})z^{-i} \Big\|_{\infty}^2 \leq    \sum \limits_{i=0}^{\infty} \Big\|  ({{G}^{\mathtt{md}}_i} - {{G}^{\mathtt{pt}}_i})z^{-i} \Big\|_{\infty}^2  $  (due to the convexity of $\|.\|_{\infty}^2$).
Then by using the two inequalities above it follows that  $\Big\| {\bf R}^{\mathtt{md}}  - {\bf R}^{\mathtt{pt}} \Big\|_{\infty}^2 \leq   \sum \limits_{i=0}^{\infty} \Big\|  ({{G}^{\mathtt{md}}_i} - {{G}^{\mathtt{pt}}_i}) \Big\|_{2}^2  = \sum \limits_{k=0}^\infty \| {C}^{\mathtt{md}}({A}^{\mathtt{md}})^k {B}^{\mathtt{md}} - {C}^{\mathtt{pt}} ({A}^{\mathtt{pt}})^k {B}^{\mathtt{pt}} \|_{2}^2 \leq \Big\| \mathcal{\hat{H}}_{0,d,d} - \mathcal{\hat{H}}_{0,\infty,\infty}\Big\|_2^2$.\\
Finally by using \eqref{probabilitybound}, 
$\Big\| {\bf R}^{\mathtt{md}}  - {\bf R}^{\mathtt{pt}} \Big\|_{\infty}^2 \leq 144 c^2 \beta^2 \mathcal{R}^2 \bigg({\dfrac{m\hat{d}+p\hat{d}^2+ \hat{d}log(T/\delta)}{T}}\bigg)$.
\end{proof}

\end{document}